\documentclass[10pt, a4paper]{article}
\usepackage{color,url,latexsym,amssymb,amsthm,amsmath,amsxtra,amsfonts}
\usepackage{latexsym,amssymb,amsthm,amsmath,amsxtra,amsfonts}
\usepackage{amsfonts}
\usepackage{amsbsy}
\usepackage{amssymb}
\usepackage{amsmath}
\usepackage{amsmath,amscd}

\theoremstyle{plain}
\newtheorem{theorem}{Theorem}[section]

\newtheorem{proposition}{Proposition}[section]
\theoremstyle{definition}

\newtheorem{example}{Example}[section]

\begin{document}

\title{Nonlinear constraints in nonholonomic mechanics}
\author{Paul Popescu and Cristian Ida}
\date{}
\maketitle

\begin{abstract}
 In this paper we have  obtained some dynamics equations, in the presence of
nonlinear nonholonomic constraints and according to a lagrangian and some
Chetaev-like conditions. Using some natural regular conditions, a simple
form of these equations is given. In the particular cases of linear and
affine constraints, one recovers the classical equations in the forms known
previously, for example, by Bloch and all \cite {Bl, BKMM}. The case of time-dependent
constraints is also considered. Examples of linear constraints, time
independent and time depenndent nonlinear constraints are considered, as
well as their dynamics given by suitable lagrangians. All  examples are
based on classical ones, such as those given by Appell's machine.
\end{abstract}

\begin{flushleft}
\strut \textbf{2010 Mathematics Subject Classification:} 70F25, 37J60, 70H45

\end{flushleft}

\section{Introduction}

The geometrization of nonholonomic systems is a 
 historically outstanding
problem in mechanics and geometry (see, for example \cite{Leon}). In
general, the most  frequently  used and studied constraints in nonholonomic mechanics are
 the linear and affine ones (see, for example, \cite{Be1, Bl, BKMM, CILD, CLMM,
GG, GUG, LDV, Le, ML, SCS}). But nonlinear constraints are also involved in
nonholonomic mechanics (see, for example, \cite{Ben, Da, GLY, ILMD, KO, Kr,
Kr01, LMD, LiBe, Ma0, Ma}).

The study of equations of dynamics ruled by a lagrangian and some nonlinear
constraints is usually associated with Chetaev or generalized Chetaev
principles.  Some criticism concerning Chetaev's principle is  drawn, for
example, in \cite{Ma, Ma0}, where some situations (as Appell machine) are
presented as examples when Chetaev principle fails to a real situation.
Other authors use Chetaev principle in some special conditions, as for
example in \cite{LiBe}, as a generalized Chetaev principle, when the
constraint is homogeneous in the relative velocities and the constraints are
time dependent. Our goal in the paper is not to study the workability of
Chetaev or generalized Chetaev principle, but the possibility to put in an
unitary form the dynamics equations coming from linear, affine and regular
nonlinear constraints (Theorem \ref{thmain}).

The Chetaev principle, generally accepted in nonlinear constraint case,
comes from the following principle: take the variation before imposing the
constraints, that is, not imposing the constraints on the family of curves
defining the variation. In this case we follow similar arguments as in the
linear and affine constraints in \cite{Bl, BKMM} and we give a new form
expressed in Theorem \ref{thmain}. Adapting these results in the case of
time dependent nonlinear constraints, we obtain a similar general result
that applies in the cases of generalized Chetaev case \cite[Section 2]{LiBe}
or the example in \cite[Section 3]{LiBe}.

Some short preliminaries on foliations are given in the second section.
Nonlinear constraints, including linear and affine ones, are considered for
lagrangians in the third section using foliations, but following the
classical bundle setting as in \cite{Bl} for linear and affine constraints.
The implicit forms of constraints and a link with the Lagrange multipliers
form of Euler-Lagrange equation are also considered.

Given a nonlinear constraint $C$, then an almost transverse semi-spray gives
a $C$--semispray $S$ (Proposition \ref{prextsemisptr}) and a $C$--semispray $%
S$ gives rise to an $S$\emph{--curvature }$R$ of $C$ (Proposition \ref%
{prdefR}). An $S$\emph{--curvature }$R$ of $C$ is considered in the paper,
in the case of nonlinear constraints, since in the cases of linear and
affine constraints, the curvature $R$ of $C$ appears naturally defined \cite%
{Bl, BKMM}. A short form of nonholonomic lagrangian dynamics, subject to
linear and affine constraints, are presented in the third section, following 
\cite{Bl}; this is extended in the case of dynamics generated by $C$%
--regular Lagrangians, having nonlinear constraint systems. The main result
is Theorem \ref{thmain}, where a synthetic form of linear and
regular-nonlinear cases is given. This result can be adapted to other
situations; for example, in the case of time dependent constraints, or of a
time independent lagrangian (as studied in \cite{LiBe}, see Examples \ref{e3} and \ref{e6}
in the last section). In order to illustrate the constructions performed in
the paper, some examples are considered in the last section:

-- the Appell's linear constraints (as, for example, \cite{Ma}),

-- the Appell's nonlinear constraints (as, for example, \cite{LMD, ILMD}),

-- the Appell-Hammel dynamic system in an elevator (as considered in \eqref
{ex-AH}),

-- the Benenti mechanism \cite{Ben} (see also \cite{KO}),

-- the Marle servomechanism \cite{Ma0} (see also \cite{KO}),

-- a decelerated motion of a free particle in the presence of quadratic
constraints given by a riemannian flow, extending the euclidian case, as
studied in \cite{Kr, Sw}.

We use basic constructions and results on vector bundles and foliations from
classical sources \cite{La} and \cite{Mo} respectively; for the use of
vector bundles and foliations in mechanics one follow \cite{Bu-Mi} and \cite%
{Po1} respectively. We notice also that almost all formulas obtained in the paper,
  except for some explicit situations, have the same form in the simple
(i.e. fibered manifold) case, as well as in the foliated case. Throughout
the paper we pointed out the main constructions and formulas to the special
case of a simple foliation, when only the notations of geometrical objects
are different.

\section{Preliminaries on foliations}

Let us consider an $(n+m)$-dimensional manifold $M$, assumed to be connected
and orientable.

A codimension ${n}$ foliation $\mathcal{F}$ on $M$ is defined by a foliated
cocycle $\{U_{i},\varphi _{i},f_{i,j}\}$ such that:

\begin{enumerate}
\item[(i)] $\{U_i\}$, $i\in I$ is an open covering of $M$;

\item[(ii)] For every $i\in I$, $\varphi _{i}:U_{i}\rightarrow T$ are
submersions, where $T$ is an $n$-dimensional manifold, called \emph{%
transversal manifold};

\item[(iii)] The maps $f_{i,j}:\varphi_i(U_i\cap U_j)\rightarrow
\varphi_j(U_i\cap U_j)$ satisfy 
\begin{equation}
\varphi_j=f_{i,j}\circ\varphi_i  \label{I1}
\end{equation}
for every $(i,j)\in I\times I$ such that $U_i\cap U_j\neq\emptyset$.
\end{enumerate}

Every fibre of $\varphi _{i}$ is called a \textit{plaque} of the foliation.
Condition ({\ref{I1}}) says that, on the intersection $U_{i}\cap U_{j}$ the
plaques defined respectively by $\varphi _{i}$ and $\varphi _{j}$  coincide.
The manifold $M$ is decomposed into a family of disjoint immersed connected
submanifolds of dimension $m$; each of these submanifolds is called a 
\textit{leaf} of $\mathcal{F}$. If $U\subset M$ is an open subset, then a
foliation $\mathcal{F}$ on $M$ induces a foliation $\mathcal{F}_{U}$ on $U$,
called an \emph{induced foliation}.

We denote by $T\mathcal{F}$ the tangent bundle to $\mathcal{F}$ and by $%
\Gamma (T\mathcal{F})$ the module of its global sections, i.e. the vector
fields on $M$ tangent to $\mathcal{F}$. The \emph{normal bundle} of $%
\mathcal{F}$ is $N\mathcal{F}=TM/T\mathcal{F}$. A vector field on $M$ is 
\emph{transverse} if it locally projects to the transversal manifold.

A system of local coordinates adapted to $\mathcal{F}$ are coordinates $%
(x^{u},x^{\bar{u}})$, $u=1,\ldots ,m$, $\bar{u}=1,\ldots ,n$ on an open
subset $U$, where $\mathcal{F}_{U}$ is trivial and defined by the equations $%
dx^{\bar{u}}=0,\,\bar{u}=1,\ldots ,{n}$.

A particular example of a foliation is a \emph{fibered manifold}, when the
leaves are the fibers of a surjective submersion $\pi:M\rightarrow M^{\prime }$%
; this is called a \emph{simple foliation}. A particular example of a
fibered manifold is a \emph{locally trivial fibration}. There are elementary
examples of simple foliations that come from no trivial fibrations and the
spaces of leaves are not Hausdorff separated. For example, considering the
natural projection $\pi _{1}:I\!\!R^{2}\rightarrow I\!\!R$, $(x,y)\overset{%
\pi _{1}}{\rightarrow }{}x$ one  obtains a foliation $\mathcal{F}$ on $%
I\!\!R^{2}$; but on $U=I\!\!R^{2}\backslash \{(x,0)|x\geq 0\}\subset
I\!\!R^{2}$ the induced foliation $\mathcal{F}_{U}$ is not a locally trivial
fibration and the space of leaves is not Hausdorff separated (even if the
leaves are fibers of a surjective submersion). According to the above
conventions, the coordinates are denoted by $x=x^{\bar{1}}$ and $y=x^{1}$.

In the case when a foliation $\mathcal{F}$ on $M$ is simple, i.e. the leaves
on $M$ are the fibers of a submersion $\pi :M\rightarrow M^{\prime }$, the
equations have the same form as in the case of a general foliation. The main
difference between  the simple and foliated cases is that in the simple case, the
transverse coordinates $(x^{\bar{u}})$ are coordinates on an open subset $%
U^{\prime }=\pi\left( U\right) \subset M^{\prime }$, where $U\subset M$
has as coordinates $(x^{u},x^{\bar{u}})$, while in the foliated case it is
possible that $\pi $ may not be in any way as a global map (i.e. on $M$),
but locally there are surjective submersions $\pi _{U}:U\rightarrow
U^{\prime }$, where $U^{\prime }\subset T$ is an open subset of the
transversal manifold. But the formulas we obtain in the paper have the same
form in both simple and foliated cases. 

\section{Linear, affine and nonlinear constraints and Lagrangians}

\subsection{Linear and affine constraints}

In this subsection we perform a simple transcription of the fiber bundle
case (as for example in \cite{Bl, BKMM}) to the foliations case.

A \emph{linear constraint} system of a foliation $\mathcal{F}$ is a left
splitting of the inclusion $T\mathcal{F}\overset{I_{0}}{\rightarrow }{}TM$.
Since there is a short exact sequence of vector bundle morphisms%
\begin{equation}
0\rightarrow T\mathcal{F}\overset{I_{0}}{\rightarrow }{}TM\overset{\Pi _{0}}{%
\rightarrow }{}N\mathcal{F}\rightarrow 0,  \label{exseq}
\end{equation}%
it follows that the existence  of a left splitting $C$ of $I_{0}$ is equivalent
 to the existence  of a right splitting $D$ of the projection $\Pi _{0}$, thus
 with an inclusion of $N\mathcal{F}$ in $TM$, via the injective morphism $D$, that
gives a decomposition 
\begin{equation*}
TM=T\mathcal{F}\oplus H\mathcal{F},
\end{equation*}%
where $H\mathcal{F}=D(N\mathcal{F})$. The \emph{curvature of }$C$ is the
bilinear map $B:\Gamma (N\mathcal{F})\times \Gamma (N\mathcal{F})\rightarrow
\Gamma (T\mathcal{F})$ given by%
\begin{equation}
B\left( X,Y\right) =C\left( \left[ D(X),D(Y)\right] \right) .
\label{curvat01}
\end{equation}%
The condition that a section $\bar{X}\in \Gamma (N\mathcal{F})$ be a \emph{%
transverse field} is that for every vector fields $X,Y\in \mathcal{X}(M)$
such that $\bar{X}=\Pi _{0}(X)$ and $Y\in \Gamma (T\mathcal{F})$, then $%
\left[ X,Y\right] \in \Gamma (T\mathcal{F})$; we say that $D(\bar{X})\in
\Gamma (D(N\mathcal{F}))$ is the \emph{horizontal lift} of $\bar{X}$. Thus if $%
\bar{X},\bar{Y}\in \Gamma (N\mathcal{F})$ are transverse, the curvature has
the form 
\begin{equation}
B\left( \bar{X},\bar{Y}\right) =C\left( \left[ \bar{X}^{h},\bar{Y}^{h}\right]
\right) .  \label{curvat02}
\end{equation}

Using local coordinates $(x^{u},x^{\bar{u}})$ on $M$ and the corresponding
ones $(y^{u},y^{\bar{u}})$ on the fibers of $TM$, a linear constraint $C$
has the local form 
\begin{equation}
(x^{u},x^{\bar{u}},y^{u},y^{\bar{u}})\overset{C}{\rightarrow }{}(x^{u},x^{%
\bar{u}},y^{u}+C_{\bar{u}}^{u}(x^{u},x^{\bar{u}})y^{\bar{u}})
\label{locflinc}
\end{equation}%
and the corresponding $D$ is%
\begin{equation}
(x^{u},x^{\bar{u}},y^{\bar{u}})\overset{D}{\rightarrow }{}(x^{u},x^{\bar{u}%
},-C_{\bar{u}}^{u}(x^{u},x^{\bar{u}})y^{\bar{u}},y^{\bar{u}}).
\label{locflind}
\end{equation}

The curvature $B$ of $C$ has the local form%
\begin{eqnarray}
B_{\bar{u}\bar{v}}^{u}\dfrac{\partial }{\partial x^{u}} &=&B\left( \frac{%
\delta }{\delta x^{\bar{u}}},\frac{\delta }{\delta x^{\bar{v}}}\right) =%
\left[ \frac{\delta }{\delta x^{\bar{u}}},\frac{\delta }{\delta x^{\bar{v}}}%
\right]   \label{eqcurvat} \\
&&=\left(\frac{\partial C_{\bar{u}}^{u}}{\partial x^{\bar{v}}}-\frac{\partial C_{%
\bar{v}}^{u}}{\partial x^{\bar{u}}}+C_{\bar{v}}^{v}\frac{\partial C_{\bar{u}%
}^{u}}{\partial x^{v}}-C_{\bar{u}}^{v}\frac{\partial C_{\bar{v}}^{u}}{%
\partial x^{v}}\right)\frac{\partial}{\partial x^u},  \notag
\end{eqnarray}%
where%
\begin{equation*}
\frac{\delta }{\delta x^{\bar{u}}}={\left( \dfrac{\partial }{\partial x^{%
\bar{u}}}\right) ^{h}}=\dfrac{\partial }{\partial x^{\bar{u}}}-C_{\bar{u}%
}^{u}(x^{u},x^{\bar{u}})\dfrac{\partial }{\partial x^{u}}.
\end{equation*}

In the case when the foliation $\mathcal{F}$ is simple, given by the fibers
of a fibered manifold $\pi :M\rightarrow M^{\prime }$, then $T\mathcal{F}$
is just the vertical bundle $VM=\ker \pi _{\ast }$, where $\pi _{\ast
}:TM\rightarrow TM^{\prime }$  is the differential map of $\pi $, and a linear
constraint $C$ is just an Ehresmann connection on $M$. The vector bundle $N%
\mathcal{F}$ is isomorphic  to the quotient vector bundle $TM/VM\overset{%
not.}{=}{}NM$; denote by $\pi _{NM}:NM\rightarrow M$ the canonical
projection. If the fibered manifold is locally trivial, then the vector
bundle $N\mathcal{F}$ is  canonically isomorphic with the induced vector bundle 
$\pi ^{\ast }TM^{\prime }$.

As an example, we consider the linear Appell constraints (see, for example, 
\cite{Ma}). The manifold is $M=I\!\!R^{3}\times T^{2}$ and the foliation is
the simple foliation defined by the fibers of the canonical projection $%
I\!\!R^{3}\times T^{2}\rightarrow T^{2}$. Consider the coordinates $%
(x^{1},x^{2},x^{3})$ on $I\!\!R^{3}$ and $(x^{\bar{1}},x^{\bar{2}})$ on $%
T^{2}$. The linear Appell constraints are given by the formulas 
\begin{equation}
C^{1}=Ry^{\bar{1}}\cos x^{\bar{2}},C^{2}=Ry^{\bar{1}}\sin x^{\bar{2}%
},C^{3}=ry^{\bar{1}}.  \label{exlinAppC}
\end{equation}%
Using formulas (\ref{eqcurvat}), its curvature $B$ has the coefficients 
\begin{equation}
B_{\bar{1}\bar{2}}^{1}=-R\sin x^{\bar{2}},B_{\bar{1}\bar{2}}^{2}=R\cos x^{%
\bar{2}},B_{\bar{1}\bar{2}}^{3}=0.  \label{CurvApp}
\end{equation}

An \emph{affine constraint} system of a foliation $\mathcal{F}$ is a fibered
map $D^{\prime }:N\mathcal{F}\rightarrow TM$ affine on fibers. One can
decompose $D^{\prime }$ as 
\begin{equation*}
D^{\prime }(\bar{X})=D(\bar{X})-b,
\end{equation*}%
where $D$ comes from a linear constraint $C:TM\rightarrow T\mathcal{F}$ and $%
b\in \Gamma (T\mathcal{F})$ is a tangent vector field to $\mathcal{F}$. We
can  also define a map $C^{\prime }:TM\rightarrow T\mathcal{F}$, by $%
C^{\prime }(X)=C(X)+b$. In the affine case,  it can be easily seen that giving $C$ and $b$ is equivalent
 to giving $D$ and $b$.

In  a similar way one can extend the definition of an adapted Lagrangian $L$%
, asking that $L$ has the form 
\begin{equation}
L\left( X\right) =L_{0}(C^{\prime }(X))+\bar{L}\left( \Pi _{0}(X)\right)
,X\in \mathcal{X}(\widetilde{TM}),  \label{locflag'}
\end{equation}%
where $C^{\prime }$ is an affine constraint and $\widetilde{TM}=TM-\{\mathrm{%
zero\,\,section\}}$.

According to \cite[Ch. 5]{Bl}, a covariant derivative of $b$, along a
horizontal vector field $\bar{X}\in \Gamma (D(N\mathcal{F}))$, can be
considered as a vector field $\nabla _{\bar{X}}b\in \mathcal{X}(M)$ that
projects by $\Pi _{0}$ on $\bar{X}$. Using local coordinates, if a linear
constraint has the local form (\ref{locflind}), then $C$ (corresponding to $%
D $) and $D^{\prime }$ have the forms (\ref{locflinc}) and 
\begin{equation}
(x^{u},x^{\bar{u}},y^{\bar{u}})\overset{D^{\prime }}{\rightarrow }%
{}(x^{u},x^{\bar{u}},b^{u}(x^{u},x^{\bar{u}})-C_{\bar{u}}^{u}(x^{u},x^{\bar{u%
}})y^{\bar{u}})  \label{locflind'}
\end{equation}%
respectively. If 
\begin{equation*}
\bar{X}=\bar{X}^{u}\left( \dfrac{\partial }{\partial x^{\bar{u}}}-C_{\bar{u}%
}^{u}\dfrac{\partial }{\partial x^{u}}\right) ,
\end{equation*}%
then 
\begin{equation*}
\left( x^{u},x^{\bar{u}}\right) \overset{\nabla _{\bar{X}}b}{\rightarrow }%
{}\left( x^{u},x^{\bar{u}},\bar{X}^{\bar{v}}\gamma _{\bar{v}}^{u},\bar{X}^{%
\bar{u}}\right) ,
\end{equation*}%
where%
\begin{equation*}
\gamma _{\bar{u}}^{u}=\frac{\partial b^{u}}{\partial x^{\bar{u}}}-C_{\bar{u}%
}^{v}\frac{\partial b^{u}}{\partial x^{v}}+b^{v}\frac{\partial C_{\bar{u}%
}^{u}}{\partial x^{v}}.
\end{equation*}

Using nonlinear constraints approach, studied in the next sections, the
curvature $R$ of $D$ can be expressed also by the formula%
\begin{equation*}
R_{\bar{u}}^{u}\dfrac{\partial }{\partial x^{u}}=-\left[ \left[ \dfrac{%
\partial }{\partial y^{\bar{u}}},C_{V}\right] ,C_{V}\right] =\left( \gamma _{%
\bar{u}}^{u}+B_{\bar{u}\bar{v}}^{u}y^{\bar{v}}\right) \dfrac{\partial }{%
\partial x^{u}},
\end{equation*}%
where $C_{V}=\left( C_{\bar{u}}^{u}y^{\bar{u}}-b^{u}\right) \dfrac{\partial 
}{\partial x^{u}}+y^{\bar{u}}\dfrac{\partial }{\partial x^{\bar{u}}}$.
According to Proposition \ref{prdefR} in the next section, the curvature is
represented by a global tensor $R\in L(VN\mathcal{\mathcal{F}},T\mathcal{F}%
_{N\mathcal{F}}\mathcal{)}$, while $C_{V}$ is not a tensor, having only a
local vector field form. This curvature and the covariant derivative $\nabla 
$ play an important role in \cite[Sect. 5.2]{Bl} to express the nonholonomic
equations of motion, in the case of linear and affine constraints.

\subsection{Non-linear constraints}

In this subsection we extend the linear and affine constraints studied in
the previous subsection, to nonlinear constraints.

Let us consider the endomorphism $\tilde{J}\in End\left( TN\mathcal{F}%
\right) $, induced by the projection of the canonical almost tangent
structure $J\in End\left( TTM\right) $. Let $VN\mathcal{F}$ be the vertical
vector bundle of $N\mathcal{F}$ and $\Gamma _{0}\in \Gamma (VN\mathcal{F})$
be the transverse Liouville vector field. Using local coordinates, 
\begin{equation}
\dfrac{\partial }{\partial x^{u}}\overset{\tilde{J}}{\rightarrow }{}0,\dfrac{%
\partial }{\partial x^{\bar{u}}}\overset{\tilde{J}}{\rightarrow }{}\dfrac{%
\partial }{\partial y^{\bar{u}}},\dfrac{\partial }{\partial y^{\bar{u}}}%
\overset{\tilde{J}}{\rightarrow }{}0,\;\Gamma _{0}=y^{\bar{u}}\dfrac{%
\partial }{\partial y^{\bar{u}}}  \label{projpi}
\end{equation}%
and the local sections of $\Gamma (VN\mathcal{F})$ are spanned by $\{\frac{%
\partial }{\partial y^{\bar{u}}}\}$.

We say that a map $C:N\mathcal{F}\rightarrow TM$, viewed also as a section $%
C\in \Gamma (\pi _{N\mathcal{F}}^{\ast }TM)$, is a \emph{nonlinear constraint%
} if $\tilde{J}\left( C\right) =\Gamma _{0}$. Using local coordinates,%
\begin{equation}
(x^{u},x^{\bar{u}},y^{\bar{u}})\overset{C}{\rightarrow }{}(x^{u},x^{\bar{u}%
},C^{u}(x^{v},x^{\bar{v}},y^{\bar{v}}),y^{\bar{u}}),C=C^{u}\dfrac{\partial }{%
\partial x^{u}}+y^{\bar{u}}\dfrac{\partial }{\partial x^{\bar{u}}}.
\label{eqdefC}
\end{equation}

\begin{proposition}
\label{prconsdoi}A nonlinear constraint  gives rise to a left splitting $%
C^{\prime \prime }$ or, equivalently,  to a right splitting $D^{\prime \prime }$
of the exact sequence of vector bundle morphisms%
\begin{equation}
0\rightarrow \pi _{N\mathcal{F}}^{\ast }T\mathcal{F}\overset{I_{0}^{\prime
\prime }}{\rightarrow }{}\pi _{N\mathcal{F}}^{\ast }TM\overset{\Pi
_{0}^{\prime \prime }}{\rightarrow }{}\pi _{N\mathcal{F}}^{\ast }N\mathcal{F}%
\rightarrow 0.  \label{exseq011}
\end{equation}
\end{proposition}

\begin{proof} Using local coordinates, it can be proved that the map 
\begin{equation}
X^{u}\dfrac{\partial }{\partial x^{u}}+X^{\bar{u}}\dfrac{\partial }{\partial
x^{\bar{u}}}\overset{C^{\prime }}{\rightarrow }{}\left( X^{u}+\dfrac{%
\partial C^{u}}{\partial y^{\bar{u}}}X^{\bar{u}}\right) \dfrac{\partial }{%
\partial x^{u}}  \label{formDsec}
\end{equation}%
gives a left splitting of $I_{0}^{\prime \prime }$. 
\end{proof}

It follows that there is an inclusion of $\pi _{N\mathcal{F}}^{\ast }N%
\mathcal{F}$ in $\pi _{N\mathcal{F}}^{\ast }TM$, via the injective morphism $%
D^{\prime \prime }$, that gives a decomposition 
\begin{equation}
\pi _{N\mathcal{F}}^{\ast }TM=\pi _{N\mathcal{F}}^{\ast }T\mathcal{F}\oplus
N^{\prime \prime }\mathcal{F},  \label{nonldecomp}
\end{equation}%
where $N^{\prime \prime }\mathcal{F}=D^{\prime \prime }(\pi _{N\mathcal{F}%
}^{\ast }N\mathcal{F})$.

Assume  now that the foliation $\mathcal{F}$ is simple and the leaves are the
fibers of a fibered manifold $\pi :M\rightarrow M^{\prime }$. Then the lifts
of the Liouville vector field and the almost tangent structure on $M^{\prime
}$ to the vertical vector bundle $VNM=V(NM)$ of $NM=TM/VM$ are $\Gamma _{0}$
and $\tilde{J}$ respectively. A nonlinear constraint is a fibered manifold
map $C:VM\rightarrow TM$, viewed also as a section $C\in \Gamma (\pi
_{NM}^{\ast }TM)$, such that $\tilde{J}\left( C\right) =\Gamma _{0}$. Notice
also that the exact sequence of vector bundle morphisms (\ref{exseq011})
have the following form%
\begin{equation}
0\rightarrow \pi _{NM}^{\ast }VM\overset{I_{0}^{\prime \prime }}{\rightarrow 
}{}\pi _{NM}^{\ast }TM\overset{\Pi _{0}^{\prime \prime }}{\rightarrow }{}\pi
_{NM}^{\ast }NM\rightarrow 0.  \label{exseq011s}
\end{equation}

We deal now with an implicit realization of nonlinear constraints\emph{.}

Let $F:TM\rightarrow T\mathcal{F}$ be a fibered manifold map over the base $%
M $. Using local coordinates, $F$ has the form%
\begin{equation*}
(x^{u},x^{\bar{u}},y^{u},y^{\bar{u}})\overset{F}{\rightarrow }{}(x^{u},x^{%
\bar{u}},F^{u}(x^{v},x^{\bar{v}},y^{v},y^{\bar{v}})).
\end{equation*}%
Let us notice that the property of a point $z\in TM$, of coordinates $%
(x^{v},x^{\bar{v}},y^{v},y^{\bar{v}})$, to have $F^{u}(x^{v},x^{\bar{v}%
},y^{v},y^{\bar{v}})=0$, does not depend on coordinates; let us denote as $%
\mathcal{C}_{F}$ the set of these points.

We also say that $F$ is a \emph{contravariant implicit constraint }(or a 
\emph{con--constraint} for short) if

\begin{enumerate}
\item for every $x\in M$ and any transverse vector $\bar{X}_{x}\in N_{x}%
\mathcal{F}$, there is a point in $T_{x}M\cap \mathcal{C}_{F}$ that projects
on $\bar{X}_{x}$;

\item the local matrices $\left( \dfrac{\partial F^{u}}{\partial y^{v}}%
(z)\right) $ are non-singular in all $z\in \mathcal{C}_{F}$.
\end{enumerate}

By the implicit mapping theorem and using local coordinates, these
conditions can be read that the local equations $F^{u}(x^{v},x^{\bar{v}%
},y^{v},y^{\bar{v}})=0$ can be solved with respect to $y^{v}$, giving local
functions $(x^{u},x^{\bar{u}},y^{\bar{u}})\rightarrow C^{u}(x^{u},x^{\bar{u}%
},y^{\bar{u}})$ in a neighborhood of any point in $N\mathcal{F}$, such that $%
F^{u}(x^{v},x^{\bar{v}},C^{v},y^{\bar{v}})=0$. Finally, we obtain local
nonlinear constraints $C_{U}:N\mathcal{F}_{U}\rightarrow TU$, where $%
U\subset M$ are open sets that cover $M$.

Consider now the covariant case.

Let $G:TM\rightarrow T^{\ast }\mathcal{F}$ be a fibered manifold map over
the base $M$. Using local coordinates, $G$ has the form%
\begin{equation*}
(x^{u},x^{\bar{u}},y^{u},y^{\bar{u}})\overset{G}{\rightarrow }{}(x^{u},x^{%
\bar{u}},G_{u}(x^{v},x^{\bar{v}},y^{v},y^{\bar{v}})).
\end{equation*}%
As in the contravariant case, the property of a point $z\in TM$, that the
coordinates $(x^{v},x^{\bar{v}},y^{v},y^{\bar{v}})$ fulfill $G_{u}(x^{v},x^{%
\bar{v}},y^{v},y^{\bar{v}})=0$, does not depend on coordinates; we denote
also by $\mathcal{C}_{G}$ the set of these points.

We say that $G$ is a \emph{covariant implicit constraint }(or a \emph{%
cov--constraint} for short) if

\begin{enumerate}
\item for every $x\in M$ and any transverse vector $\bar{X}_{x}\in N_{x}%
\mathcal{F}$, there is a point in $T_{x}M\cap \mathcal{C}_{G}$ that projects
on $\bar{X}_{x}$;

\item the local matrices $\left( \dfrac{\partial G_{u}}{\partial y^{v}}%
(z)\right) $ are non-singular in all $z\in \mathcal{C}_{G}$.
\end{enumerate}

These conditions can be read that the local equations $G_{u}(x^{v},x^{\bar{v}%
},y^{v},y^{\bar{v}})=0$ can be solved with respect to $y^{u}$, giving local
functions $(x^{u},x^{\bar{u}},y^{\bar{u}})\rightarrow C^{u}(x^{u},x^{\bar{u}%
},y^{\bar{u}})$ in a neighborhood of any point in $N\mathcal{F}$, such that $%
G_{u}(x^{v},x^{\bar{v}},C^{v},y^{\bar{v}})=0$. Finally, as in the
contravariant case, we obtain local nonlinear constraints $C_{U}:N\mathcal{F}%
_{U}\rightarrow TU$, where $U\subset M$ are open sets that cover $M$.

The implicit form of constraints can be used to give an invariant form to
the condition that a covector type be a combination of partial derivatives
of functions that give the constraints; for example, in nonholonomic
mechanics, the Chetaev condition reads that the covector giving the
Euler-Lagrange derivative is such a combination.

\begin{proposition}
\label{primpl}Let $E:T\mathcal{F}\rightarrow T^{\ast }M$ be a fibered
manifold map over the base $M$. If $E=E_{u}dx^{u}+E_{\bar{u}}dx^{\bar{u}}$
has the property%
\begin{equation*}
E_{u}=\sum\limits_{v}\lambda _{v}\dfrac{\partial F^{v}}{\partial y^{u}}%
(x^{u},x^{\bar{u}},C^{u},y^{\bar{u}}),E_{\bar{u}}=\sum\limits_{v}\lambda _{v}%
\dfrac{\partial F^{v}}{\partial y^{\bar{u}}}(x^{u},x^{\bar{u}},C^{u},y^{\bar{%
u}})
\end{equation*}%
on $\mathcal{C}_{F}$,  for a given con-constraint $F$, or%
\begin{equation*}
E_{u}=\sum\limits_{v}\lambda ^{v}\dfrac{\partial F_{v}}{\partial y^{u}}%
(x^{u},x^{\bar{u}},C^{u},y^{\bar{u}}),E_{\bar{u}}=\sum\limits_{v}\lambda ^{v}%
\dfrac{\partial F_{v}}{\partial y^{\bar{u}}}(x^{u},x^{\bar{u}},C^{u},y^{\bar{%
u}})
\end{equation*}%
on $\mathcal{C}_{F}$,   for a given cov-constraint $F$, then the following
identity holds true: 
\begin{equation*}
E_{\bar{u}}+\sum\limits_{u}\dfrac{\partial C^{u}}{\partial y^{\bar{u}}}%
E_{u}=0.
\end{equation*}
\end{proposition}

\begin{proof} We consider the cov-constraint case, since the con-constraint
case is analogous. Differentiating the implicit equation $F_{u}(x^{u},x^{%
\bar{u}},C^{u},y^{\bar{u}})=0$ with respect to $y^{\bar{u}}$, we obtain 
\begin{equation*}
\sum\limits_{v}\dfrac{\partial C^{v}}{\partial y^{\bar{u}}}\dfrac{\partial
F_{u}}{\partial y^{v}}(x^{u},x^{\bar{u}},C^{u},y^{\bar{u}})+\dfrac{\partial
F_{u}}{\partial y^{\bar{u}}}(x^{u},x^{\bar{u}},C^{u},y^{\bar{u}})=0,
\end{equation*}%
thus using the hypothesis, the conclusion follows. 
\end{proof}

Nonlinear constraints  lift to linear constraints of the natural lifted
foliation $\mathcal{F}_{N\mathcal{F}}$ on $N\mathcal{F}$, as follows. On an
intersection of two adapted charts, the rule is%
\begin{equation}
C^{u^{\prime }}(x^{v^{\prime }},x^{\bar{v}^{\prime }},y^{\bar{v}^{\prime }})=%
\frac{\partial x^{u^{\prime }}}{\partial x^{u}}C^{u}(x^{v},x^{\bar{v}},y^{%
\bar{v}})+\frac{\partial x^{u^{\prime }}}{\partial x^{\bar{u}}}y^{\bar{u}}.
\label{eqdefC2}
\end{equation}%
Using this formula, by a direct computation, one can check that the formulas 
$C_{\bar{u}}^{u}=\frac{\partial C^{u}}{\partial y^{\bar{u}}}$, $%
C_{v}^{u}=0$  give rise to a linear constraint on $\mathcal{F}_{N%
\mathcal{F}}$, i.e. a splitting (left $\mathcal{C}$ or right $\mathcal{D}$)
of the exact sequence 
\begin{equation}
0\rightarrow T\mathcal{F}_{N\mathcal{F}}\overset{I_{0}^{\prime }}{%
\rightarrow }{}T(N\mathcal{F)}\overset{\Pi _{0}^{\prime }}{\rightarrow }{}N%
\mathcal{F}_{N\mathcal{F}}\rightarrow 0.  \label{exseq02}
\end{equation}

In the case when the foliation $\mathcal{F}$ is simple, that is, the leaves are
the fibers of a fibered manifold $\pi :M\rightarrow M^{\prime }$, a
con-constraint is defined by a fibered manifold map $F:TM\rightarrow VM=\ker
\pi _{\ast }$, while a cov-constraint is defined by a fibered manifold map $%
G:TM\rightarrow V^{\ast }M$.

If $C:TM\rightarrow T\mathcal{F}$ is a linear constraint, it is a nonlinear
one as well. Indeed, the right splitting $D:N\mathcal{F}\rightarrow TM$
induces by $\pi _{N\mathcal{F}}^{\ast }D:\pi _{N\mathcal{F}}^{\ast }N%
\mathcal{F}\rightarrow \pi _{N\mathcal{F}}^{\ast }TM$, the vector field $%
C^{\prime }=\pi _{N\mathcal{F}}^{\ast }D(\Gamma _{0})$ that is a nonlinear
constraint.

An affine constraint gives rise also to a nonlinear one, in a similar way.
Indeed, an affine constraint is given by a linear constraint $C$ and a
vector field $b\in \Gamma \left( T\mathcal{F}\right) $. The vector field $%
C^{\prime }=\pi _{N\mathcal{F}}^{\ast }D(\Gamma _{0})+\pi _{N\mathcal{F}%
}^{\ast }b$ gives a nonlinear constraint, where $D$ is the right splitting
of (\ref{exseq}) corresponding to $C$.

According to \cite{Bl, BKMM} (see also the next section) some linear or
affine constraints, give rise to a curvature, which is a tensor. In order to
study the case of nonlinear constraints, we have to consider instead a
semi-spray, which is  not a tensor anymore.

An \emph{almost transverse semi-spray} is a (non-necessarily foliated)
section $S:N\mathcal{F}\rightarrow NN\mathcal{F}$ that is a section for both
vector bundle structures on $NN\mathcal{F}$ over $N\mathcal{F}$ (one of
usual vector bundle, the other one induced by the transversal component of
the differential of the canonical projection $N\mathcal{F}\rightarrow M$, as
a foliated map). In the case of the trivial foliation by the points of $M$,
we recover the definition of a semi-spray on $M$.

In the case of a simple foliation given by $\pi :M\rightarrow M^{\prime }$,
an almost transverse semi-spray is a section $S:NM\rightarrow N(NM)=NNM$
that is a section for both vector bundle structures on $NNM$ over $NM$ (one
of usual vector bundle, the other one induced by the projection of the
differential of the canonical projection $NM\rightarrow M$, as a fibered
manifold map over $M^{\prime }$).

Using local coordinates, an almost transverse semi-spray $S$ has the local
form%
\begin{equation}
\left( x^{u},x^{\bar{u}},y^{\bar{u}}\right) \overset{S}{\rightarrow }%
{}\left( x^{u},x^{\bar{u}},y^{\bar{u}},y^{\bar{u}},S^{\bar{u}}(x^{u},x^{\bar{%
u}},y^{\bar{u}})\right) .  \label{locfS}
\end{equation}%
 We say that $S$ is a transverse
semi-spray if $S$ happens to be a foliate section. This condition means that in formula (\ref{locfS}) one have $S^{%
\bar{u}}=S^{\bar{u}}(x^{\bar{u}},y^{\bar{u}})$.~

In order to lift an (almost) transverse semi-spray one needs a nonlinear
constraint (in particular it can be a linear or an affine one).

\begin{proposition}
\label{prextsemisptr}If $S\in \Gamma (N\widetilde{N\mathcal{F}})$ is an
almost transverse semi-spray and $C:N\mathcal{F}\rightarrow TM$ is a
nonlinear constraint, then there is a unique vector field $\mathcal{S}\in 
\mathcal{X}(\widetilde{N\mathcal{F}})$ that projects by $T\widetilde{N%
\mathcal{F}}\rightarrow N\widetilde{N\mathcal{F}}$ and $T\widetilde{N%
\mathcal{F}}\rightarrow TM$ to $S$ and $C$ respectively.
\end{proposition}

\begin{proof} We use local coordinates $\left( x^{u},x^{\bar{u}},y^{\bar{u}%
}\right) $ on an open set $V=\pi _{N\mathcal{F}}^{-1}U\subset N\mathcal{F}$,
corresponding to some coordinates $\left( x^{u},x^{\bar{u}}\right) $ on $%
U\subset M$.  Consider $S$ and $C$ having the local forms (\ref{locfS}) and (%
\ref{eqdefC}) respectively. Taking into account the conditions, then $%
\mathcal{S}$ has the local form%
\begin{eqnarray}
\mathcal{S} &=&C^{u}\dfrac{\partial }{\partial x^{u}}+y^{\bar{u}}\dfrac{%
\partial }{\partial x^{\bar{u}}}+S^{\bar{u}}\dfrac{\partial }{\partial y^{%
\bar{u}}}  \label{eqlocsemisp} \\
&=&C_{V}+S^{\bar{u}}\dfrac{\partial }{\partial y^{\bar{u}}}.  \notag
\end{eqnarray}%
By a straightforward verification of chain rules on the intersection
domains, one can check that $\mathcal{S}$ is a global vector field. 
\end{proof}

Notice that considering coordinates $\left( x^{u},x^{\bar{u}},y^{\bar{u}%
}\right) $ and $\left( x^{u^{\prime }},x^{\bar{u}^{\prime }},y^{\bar{u}%
^{\prime }}\right) $ on $V=\pi _{N\mathcal{F}}^{-1}U$ and $V^{\prime }=\pi
_{N\mathcal{F}}^{-1}U^{\prime }$ respectively, then 
\begin{eqnarray}
C_{V} &=&C_{V^{\prime }}+y^{\bar{v}}\dfrac{\partial x^{\bar{u}^{\prime }}}{%
\partial x^{\bar{v}}}\dfrac{\partial }{\partial y^{\bar{u}^{\prime }}},
\label{formC} \\
S^{\bar{u}^{\prime }}\left( x^{v^{\prime }},x^{\bar{v}^{\prime }},y^{\bar{v}%
^{\prime }}\right) &=&S^{\bar{u}}\left( x^{v},x^{\bar{v}},y^{\bar{v}}\right) 
\dfrac{\partial x^{\bar{u}^{\prime }}}{\partial x^{\bar{u}}}+y^{\bar{v}}%
\dfrac{\partial y^{\bar{u}^{\prime }}}{\partial x^{\bar{v}}}.  \label{formS}
\end{eqnarray}

A vector field $\mathcal{S}\in \mathcal{X}(N\mathcal{F})$ given by
Proposition \ref{prextsemisptr} will be called a $C$\emph{-semispray}.

Let us notice that $C_{V}$ and $C$ have the same formulas, but they are
different as vector fields; $C:N\mathcal{F}\rightarrow TM$, but $C_{V}\in 
\mathcal{X}(V)=\mathcal{X}(N\mathcal{F}_{U})$ is a local vector field.

A  historically representative nonlinear example is Appell's nonlinear
constraint,  defined as follows. Take the foliation of $I\!\!R_{0}^{3}=I\!\!R^{3}%
\backslash \{\bar{0}\}$ generated by $\dfrac{\partial }{\partial z}$. Denote 
$x=x^{\bar{1}}$, $y=x^{\bar{2}}$ and $z=x^{1}$ and consider the nonlinear
constraint given by the implicit equation 
\begin{equation}
\alpha ^{2}\left( \left( y^{\bar{1}}\right) ^{2}+\left( y^{\bar{2}}\right)
^{2}\right) -\left( y^{1}\right) ^{2}=0,\alpha \neq 0.  \label{con-nelApp}
\end{equation}

We have $y^{1}=C^{1}\left( y^{\bar{1}},y^{\bar{2}}\right) =\pm \alpha \sqrt{\left(
y^{\bar{1}}\right) ^{2}+\left( y^{\bar{2}}\right) ^{2}}$, but we take $%
C^{1}\left( y^{\bar{1}},y^{\bar{2}}\right) =\alpha \sqrt{\left( y^{\bar{1}%
}\right) ^{2}+\left( y^{\bar{2}}\right) ^{2}}$.

Formula (\ref{formDsec}) gives 
\begin{equation*}
X^{1}\dfrac{\partial }{\partial x^{1}}+X^{\bar{1}}\dfrac{\partial }{\partial
x^{\bar{1}}}+X^{\bar{2}}\dfrac{\partial }{\partial x^{\bar{2}}}\overset{%
C^{\prime }}{\rightarrow }{}\left( X^{1}+\alpha \dfrac{X^{\bar{1}}y^{\bar{1}%
}+X^{\bar{2}}y^{\bar{2}}}{\sqrt{\left( y^{\bar{1}}\right) ^{2}+\left( y^{%
\bar{2}}\right) ^{2}}}\right) \dfrac{\partial }{\partial x^{1}}.
\end{equation*}

We can consider time dependent constraints, as follows. Let us consider $%
N^{T}\mathcal{F}=N\mathcal{F}\times I\!\!R$ or $N^{T}\mathcal{F}=N\mathcal{F}%
\times S^{1}$ and the foliation $\overline{\mathcal{F}}^{T}$ on $N^{T}%
\mathcal{F}$ is induced by the foliation $\overline{\mathcal{F}}=\mathcal{F}%
_{N\mathcal{F}}$ on $N\mathcal{F}$, such that the canonical projection $N^{T}%
\mathcal{F}\rightarrow N\mathcal{F}$ is a diffeomorphism of leaves, thus the
new parameter is transverse.

A \emph{time dependent nonlinear constraint} on $M$ is a map $C:N^{T}%
\mathcal{F}\rightarrow TM$, viewed also as a section $C\in \Gamma (\pi
_{N^{T}\mathcal{F}}^{\ast }TM)$, such that $\tilde{J}\left( C\right) =\Gamma
_{0}$. Using local coordinates,%
\begin{equation}
(x^{u},x^{\bar{u}},y^{\bar{u}},t)\overset{C}{\rightarrow }{}(C^{u}(x^{v},x^{%
\bar{v}},y^{\bar{v}},t),y^{\bar{u}}),C=C^{u}\dfrac{\partial }{\partial x^{u}}%
+y^{\bar{u}}\dfrac{\partial }{\partial x^{\bar{u}}}.  \label{eqdefCa}
\end{equation}

There is an exact sequence, induced by (\ref{exseq011}): 
\begin{equation}
0\rightarrow \pi _{N^{T}\mathcal{F}}^{\ast }T\mathcal{F}\overset{%
I_{0}^{\prime \prime }}{\rightarrow }{}\pi _{N^{T}\mathcal{F}}^{\ast }TM%
\overset{\Pi _{0}^{\prime \prime }}{\rightarrow }{}\pi _{N^{T}\mathcal{F}%
}^{\ast }N\mathcal{F}\rightarrow 0.  \label{exseq011a}
\end{equation}%
As in the time independent case, a time dependent nonlinear constraint on $M$
gives also rise to a left splitting $C^{\prime \prime }$ or, equivalently, a
right splitting $D^{\prime \prime }$ of the exact sequence (\ref{exseq011a}%
). Analogous formulas to (\ref{formDsec}) and (%
\ref{nonldecomp}) can be obtained in local coordinates.

A more general approach of time dependent constraints can be considered
 by taking $M^{\prime }=M\times I\!\!R$ instead of $M$ and the parameter from $%
I\!\!R$ being transverse. Then transverse coordinates get $x^{\bar{u}}$,
where $\bar{u}=\overline{0,n}$ and $x^{\bar{0}}=t\in I\!\!R$. The case
considered above is when $y^{\bar{0}}=1$, corresponding to $(t,1)\equiv 
\dfrac{\partial }{\partial t},$ the tangent vector to curve $t\rightarrow t$
in $I\!\!R$. We do not use this general situation in the paper.

In the case when the foliation $\mathcal{F}$ is simple, given by a fibered
manifold $\pi :M\rightarrow M^{\prime }$, then $N^{T}\mathcal{F}=VM\times
I\!\!R$ or $N^{T}\mathcal{F}=VM\times S^{1}$ and a time dependent nonlinear
constraint on $M$ is a map $C:VM\times I\!\!R\rightarrow TM$, viewed also as
a section $C\in \Gamma (\pi _{VM\times I\!\!R}^{\ast }TM)$, such that $%
\tilde{J}\left( C\right) =\Gamma _{0}$.

A classical example of time dependent nonlinear constraint is the
Appell-Hammel dynamic system in an elevator considered in \cite{LiBe},
having the time dependent constraints 
\begin{equation}
\alpha ^{2}\left( \left( y^{\bar{1}}\right) ^{2}+\left( y^{\bar{2}}\right)
^{2}\right) -\left( y^{1}-v^{0}(t)\right) ^{2}=0.  \label{ex-AH}
\end{equation}%
It is easy to see that the above Appell example corresponds to the
particular case when $v^{0}(t)=0$.

We have $y^{1}=C^{1}\left( y^{\bar{1}},y^{\bar{2}}\right) =v^{0}(t)\pm
\alpha \sqrt{\left( y^{\bar{1}}\right) ^{2}+\left( y^{\bar{2}}\right) ^{2}}$%
; we take $C^{1}\left( y^{\bar{1}},y^{\bar{2}}\right) =v^{0}(t)+\alpha \sqrt{%
\left( y^{\bar{1}}\right) ^{2}+\left( y^{\bar{2}}\right) ^{2}}$.

Formula (\ref{formDsec}) gives 
\begin{equation*}
X^{1}\dfrac{\partial }{\partial x^{1}}+X^{\bar{1}}\dfrac{\partial }{\partial
x^{\bar{1}}}+X^{\bar{2}}\dfrac{\partial }{\partial x^{\bar{2}}}\overset{%
C^{\prime }}{\rightarrow }{}\left( X^{1}+\alpha \dfrac{X^{\bar{1}}y^{\bar{1}%
}+X^{\bar{2}}y^{\bar{2}}}{\sqrt{\left( y^{\bar{1}}\right) ^{2}+\left( y^{%
\bar{2}}\right) ^{2}}}\right) \dfrac{\partial }{\partial x^{1}}.
\end{equation*}

Notice that formula (\ref{eqdefC2}) on $T(N\mathcal{F})$ shows that 
\begin{equation*}
\dfrac{\partial x^{\bar{u}^{\prime }}}{\partial x^{\bar{u}}}\dfrac{\partial
x^{\bar{v}^{\prime }}}{\partial x^{\bar{v}}}\dfrac{\partial ^{2}C^{u^{\prime
}}}{\partial y^{\bar{u}^{\prime }}\partial y^{\bar{v}^{\prime }}}=\dfrac{%
\partial x^{u^{\prime }}}{\partial x^{u}}\dfrac{\partial ^{2}C^{u}}{\partial
y^{\bar{u}}\partial y^{\bar{v}}},
\end{equation*}
thus 
\begin{equation*}
\mathcal{C}=\dfrac{\partial ^{2}C^{u}}{\partial y^{\bar{u}}\partial y^{\bar{v%
}}}dx^{\bar{u}}\otimes dx^{\bar{v}}\otimes \dfrac{\partial }{\partial x^{u}}
\end{equation*}%
defines a tensor $\mathcal{C}\in L(VN\mathcal{\mathcal{F}}\otimes VN\mathcal{%
\mathcal{F}},TN\mathcal{F)}$. This tensor vanishes only for linear or affine
constraints. Thus a non-vanishing $\mathcal{C}$ gives a non-linear
constraint.

In both nonlinear Appell's examples, the matrix of $\mathcal{C}$ is 
\begin{equation}
\alpha\left( \left( y^{\bar{1}}\right) ^{2}+\left( y^{\bar{2}}\right) ^{2}\right)
^{-\dfrac{3}{2}}\left( 
\begin{array}{cc}
\left( y^{\bar{2}}\right) ^{2} & -y^{\bar{1}}y^{\bar{2}} \\ 
-y^{\bar{1}}y^{\bar{2}} & \left( y^{\bar{1}}\right) ^{2}%
\end{array}%
\right) .  \label{eqhuv}
\end{equation}

\begin{proposition}
\label{prdefR}If $C:N\mathcal{F}\rightarrow TM$ is a nonlinear constraint
and $S\in \mathcal{X}(N\mathcal{F})$ is a $C$--semispray, then the local
formula 
\begin{equation}
-\left[ \left[ \dfrac{\partial }{\partial y^{\bar{u}}},C_{V}\right] ,C_{V}%
\right] +S^{\bar{v}}\dfrac{\partial ^{2}C^{u}}{\partial y^{\bar{u}}\partial
y^{\bar{v}}}\dfrac{\partial }{\partial x^{u}}=R_{\bar{u}}^{u}\dfrac{\partial 
}{\partial x^{u}},  \label{defR}
\end{equation}%
gives a global tensor $R\in L(VN\mathcal{\mathcal{F}},T\mathcal{F}_{N%
\mathcal{F}}\mathcal{)}$, $R=R_{\bar{u}}^{u}\left( x^{v},x^{\bar{v}},y^{\bar{%
v}}\right) \dfrac{\partial }{\partial x^{u}}\otimes dx^{\bar{u}}$.
\end{proposition}

\begin{proof} Let us consider two coordinates systems on $V$ and $V^{\prime
} $ , $V\cap V^{\prime }\neq \emptyset $, on $N\mathcal{F}$, as in the proof
of Proposition \ref{prextsemisptr}. One can check that

\begin{equation*}
\left[ \dfrac{\partial }{\partial y^{\bar{u}}},C_{V}\right] -\dfrac{\partial
x^{\bar{u}^{\prime }}}{\partial x^{\bar{u}}}\left[ \dfrac{\partial }{%
\partial y^{\bar{u}^{\prime }}},C_{V^{\prime }}\right] =\dfrac{\partial y^{%
\bar{u}^{\prime }}}{\partial x^{\bar{u}}}\dfrac{\partial }{\partial y^{\bar{u%
}^{\prime }}},
\end{equation*}%
and we have%
\begin{equation}
-\left[ C_{V},\left[ C_{V},\dfrac{\partial }{\partial y^{\bar{u}}}\right] %
\right] =\left( C^{v}\dfrac{\partial ^{2}C^{u}}{\partial x^{v}\partial y^{%
\bar{u}}}+y^{\bar{v}}\dfrac{\partial ^{2}C^{u}}{\partial x^{\bar{v}}\partial
y^{\bar{u}}}-\dfrac{\partial C^{u}}{\partial x^{\bar{u}}}-\dfrac{\partial
C^{v}}{\partial y^{\bar{u}}}\dfrac{\partial C^{u}}{\partial x^{v}}\right) 
\dfrac{\partial }{\partial x^{u}}.  \label{formpsc}
\end{equation}

By a long and straightforward computation, one  obtains the formula%
\begin{equation*}
-\left[ \left[ \dfrac{\partial }{\partial y^{\bar{u}}},C_{V}\right] ,C_{V}%
\right] +\dfrac{\partial x^{\bar{u}^{\prime }}}{\partial x^{\bar{u}}}\left[ %
\left[ \dfrac{\partial }{\partial y^{\bar{u}^{\prime }}},C_{V^{\prime
}}^{\prime }\right] ,C_{V^{\prime }}^{\prime }\right] =y^{\bar{v}}\dfrac{%
\partial y^{\bar{v}^{\prime }}}{\partial x^{\bar{v}}}\dfrac{\partial
^{2}C^{v^{\prime }}}{\partial y^{\bar{v}^{\prime }}\partial x^{\bar{u}%
^{\prime }}}\dfrac{\partial x^{\bar{u}^{\prime }}}{\partial x^{\bar{u}}}%
\dfrac{\partial }{\partial x^{v^{\prime }}}.
\end{equation*}%
Using also formula (\ref{formC}), we obtain the conclusion. 
\end{proof}

We call $R$ given by Proposition \ref{prdefR}  the $S$\emph{--curvature}
of $C$; the definition of $R$ does not depend on $S$ only in the case when $%
\mathcal{C}=0$, i.e. when $C$ is a linear or affine constraint, as in \cite%
{BKMM, Bl}, (see the formulas (\ref{eqcurvL}) and (\ref{eqcurvA}) below). In
general, the formula 
\begin{equation*}
\dfrac{\partial }{\partial y^{\bar{u}}}\overset{R_{V}}{\rightarrow }{}-\left[
C_{V},\left[ C_{V},\dfrac{\partial }{\partial y^{\bar{u}}}\right] \right] 
\end{equation*}%
gives only a local linear map $L(VN\mathcal{\mathcal{F}}_{U},T\mathcal{F}_{N%
\mathcal{F}_{U}}\mathcal{)}$ that does not  extend to $L(VN\mathcal{\mathcal{%
F}},T\mathcal{F}_{N\mathcal{F}}\mathcal{)}$. We say that $R_{V}$ is the 
\emph{pseudo-curvature} of $C$; it is a tensor in the case of linear and
affine constraints, but in the general case of nonlinear constraints, it is
not a tensor.

In the case of Appell's nonlinear constraint, one have $R_{V}=0$, only using
some euclidean coordinates.

In the case of a simple foliation $\mathcal{F}$ given by $\pi :M\rightarrow
M^{\prime }$, we have $VN\mathcal{\mathcal{F}}=V(NM)=VNM$ and $T\mathcal{F}%
_{N\mathcal{F}}=\pi _{NM}^{\ast }VNM$. Then Proposition \ref{prdefR} asserts
that if $C:NM\rightarrow TM$ is a nonlinear constraint and $\mathcal{S}\in \mathcal{X}%
(NM)$ is a $C$--semispray, then the local formula (\ref{defR}) gives a
global tensor $R\in L(VNM,\pi _{NM}^{\ast }VM)$. Notice that the conditions
that $\mathcal{S}\in \mathcal{X}(NM)$ be a $C$--semispray is given using Proposition %
\ref{prextsemisptr}: given an almost transverse semi-spray $S\in \Gamma (NNM)
$ and a nonlinear constraint $C:NM\rightarrow TM$ , then there is a unique
vector field $\mathcal{S}\in \mathcal{X}(NM)$ that projects by $TNM\rightarrow NNM$
and $TNM\rightarrow TM$ to $S$ and $C$ respectively.

\section{The Lagrangian dynamics for linear, affine and nonlinear constraint
systems}

The goal of this section is to express the main result of the paper, i.e. to
obtain the equations of motions ruled by a lagrangian and some nonlinear
constraints regularly related to a foliation, in a similar form given in 
\cite{BKMM, Bl} for linear and affine constraints on a fiber bundle. The
cases of linear and affine constraints using foliations recover the fibre
bundle situation, when the foliation is simple, i.e. the leaves on $M$ are
the fibers of a submersion $\pi:M\rightarrow M^{\prime }$.  Moreover, we give the
expressions of some similar equations of motions in the case of time
dependent constraints.

Let $L:TM\rightarrow I\!\!R$ be a lagrangian on the total space of a
foliated manifold endowed with a system of a nonlinear (possibly linear or
affine) constraint. We study the case of nonlinear constraints, thus we
consider one given by a left splitting $\mathcal{C}$ of $I_{0}^{\prime }$,
or by a right splitting $\mathcal{D}$ of the projection $\Pi _{0}^{\prime }$
in the exact sequence (\ref{exseq02}). As in the case of linear or affine
constraints in \cite[Sect. 5.2]{Bl}, we consider that the equations of
motions governed by a lagrangian and the constraint, can be deduced imposing
the principle to apply first the variation, then the projection of the
Lagrange equations according to the constraint, adapting in this way
d'Alambert's principle. Specifically, using the decomposition (\ref%
{nonldecomp}), then the constraints effect on Lagrange equations has, as for
linear and affine constraints in \cite[Sect. 5.2]{Bl}, the forms%
\begin{eqnarray*}
\left( \frac{d}{dt}\dfrac{\partial L}{\partial y^{u}}-\dfrac{\partial L}{%
\partial x^{u}}\right) \delta x^{u}+\left( \frac{d}{dt}\dfrac{\partial L}{%
\partial y^{\bar{u}}}-\dfrac{\partial L}{\partial x^{\bar{u}}}\right) \delta
x^{\bar{u}} &=&0, \\
\delta x^{u}+C_{\bar{u}}^{u}\delta x^{\bar{u}} &=&0,
\end{eqnarray*}%
where $C_{\bar{u}}^{u}=\dfrac{\partial C^{u}}{\partial y^{\bar{u}}}$. Notice
that $\delta $ is subject to $t=const.$ Substituting $\delta x^{\bar{u}}$ in
the Lagrange equations, one obtain the induced \emph{constrained Lagrange
equations}: 
\begin{equation}
\left( \frac{d}{dt}\dfrac{\partial L}{\partial y^{\bar{u}}}-\dfrac{\partial L%
}{\partial x^{\bar{u}}}\right) -C_{\bar{u}}^{u}\left( \frac{d}{dt}\dfrac{%
\partial L}{\partial y^{u}}-\dfrac{\partial L}{\partial x^{u}}\right) =0.
\label{eqlagc}
\end{equation}

As we see below for implicit nonlinear constraints, these equations are
concordant to Chetaev conditions.

A nonlinear constraint $C\in \Gamma (\pi _{N\mathcal{F}}^{\ast }TM)$ can be
viewed as a map $C:N\mathcal{F}\rightarrow TM$, thus any lagrangian $%
L:TM\rightarrow I\!\!R$ induces by composition $N\mathcal{F}\overset{C}{%
\rightarrow }{}TM\overset{L}{\rightarrow }{}I\!\!R$ a new lagrangian $%
L_{c}=L\circ C$ on $N\mathcal{F}$, called the \emph{constrained lagrangian}:%
\begin{equation}
L_{c}(x^{u},x^{\bar{u}},y^{\bar{u}})=L(x^{u},x^{\bar{u}},C^{u},y^{\bar{u}}).
\label{eqconst01}
\end{equation}

According to \cite[Sect. 5.2]{Bl}, in the cases when%
\begin{eqnarray}
C^{u}(x^{v},x^{\bar{v}},y^{\bar{v}}) &=&y^{\bar{u}}C_{\bar{u}}^{u}(x^{v},x^{%
\bar{v}}),  \label{eqC01} \\
C^{u}(x^{v},x^{\bar{v}},y^{\bar{v}}) &=&y^{\bar{u}}C_{\bar{u}}^{u}(x^{v},x^{%
\bar{v}})+b^{u}(x^{v},x^{\bar{v}})  \label{eqC02}
\end{eqnarray}%
i.e. of linear and affine constraints respectively, the constrained Lagrange
equations (\ref{eqlagc}) can be written in terms of the constrained
lagrangian as%
\begin{equation}
\left( \frac{d}{dt}\dfrac{\partial L_{c}}{\partial y^{\bar{u}}}-\dfrac{%
\partial L_{c}}{\partial x^{\bar{u}}}\right) =C_{\bar{u}}^{u}\frac{\partial
L_{c}}{\partial x^{u}}-\frac{\partial L}{\partial y^{u}}(B_{\bar{u}\bar{v}%
}^{u}y^{\bar{v}}+\gamma _{\bar{u}}^{u}),  \label{eqconst02}
\end{equation}%
where $B^{u}_{\bar{u}\bar{v}}$ is given in \eqref{eqcurvat} and 
\begin{eqnarray}
\gamma _{\bar{u}}^{u} &=&\dfrac{\partial b^{u}}{\partial x^{\bar{u}}}+C_{%
\bar{u}}^{v}\dfrac{\partial b^{u}}{\partial x^{v}}-b^{v}\dfrac{\partial C_{%
\bar{u}}^{u}}{\partial x^{v}}  \label{formga}
\end{eqnarray}%
are both tensors, say $B$ and $\gamma $ (see \cite{BKMM, Bl} for more
details).

In the linear constraint case (i.e. $b^{u}=0$), the formula (\ref{eqcurvat})
gives, according to formula (\ref{formpsc}), that the curvature $R$ of $C$
is 
\begin{equation}
R_{\bar{u}}^{u}=y^{\bar{v}}B_{\bar{u}\bar{v}}^{u}.  \label{eqcurvL}
\end{equation}

In the affine constraint case, the formulas (\ref{eqcurvat}), (\ref{formga})
and (\ref{formpsc}) give the curvature of $C$ by 
\begin{equation}
R_{\bar{u}}^{u}=y^{\bar{v}}B_{\bar{u}\bar{v}}^{u}+\gamma _{\bar{u}}^{u}.
\label{eqcurvA}
\end{equation}

In the sequel we extend formulas (\ref{eqconst02}) to the case of nonlinear
constraints.

The equation of motion of the extended nonholonomic system is%
\begin{equation*}
-\delta L=\left( \dfrac{d}{dt}\dfrac{\partial L}{\partial y^{\bar{u}}}-%
\dfrac{\partial L}{\partial x^{\bar{u}}}+C_{\bar{u}}^{u}\left( \dfrac{d}{dt}%
\dfrac{\partial L}{\partial y^{u}}-\dfrac{\partial L}{\partial x^{u}}\right)
\right) dx^{\bar{u}}=0.
\end{equation*}%
In the case of the induced lagrangian $L_{c}$, one  has%
\begin{equation*}
-\delta L_{c}=\left( \dfrac{d}{dt}\dfrac{\partial L_{c}}{\partial y^{\bar{u}}%
}-\dfrac{\partial L_{c}}{\partial x^{\bar{u}}}-C_{\bar{u}}^{u}\dfrac{%
\partial L_{c}}{\partial x^{u}}\right) dx^{\bar{u}}.
\end{equation*}

We have, for $C_{\bar{u}}^{u}=\frac{\partial C^{u}}{\partial y^{\bar{u}}}$,

\begin{eqnarray*}
\dfrac{d}{dt}\dfrac{\partial L_{c}}{\partial y^{\bar{u}}}&=&\dfrac{d}{dt}%
\left( \dfrac{\partial L}{\partial y^{\bar{u}}}+C_{\bar{u}}^{u}\dfrac{%
\partial L}{\partial y^{u}}\right) =\dfrac{d}{dt}\dfrac{\partial L}{\partial
y^{\bar{u}}}+\dfrac{\partial L}{\partial y^{u}}\dfrac{d}{dt}C_{\bar{u}%
}^{u}+C_{\bar{u}}^{u}\dfrac{d}{dt}\dfrac{\partial L}{\partial y^{u}}, \\
\dfrac{\partial L_{c}}{\partial x^{\bar{u}}}&=&\dfrac{\partial L}{\partial
x^{\bar{u}}}+\dfrac{\partial C^{v}}{\partial x^{\bar{u}}}\dfrac{\partial L}{%
\partial y^{v}}\,,\,\,\dfrac{\partial L_{c}}{\partial x^{u}}=\dfrac{\partial
L}{\partial x^{u}}+\dfrac{\partial C^{v}}{\partial x^{u}}\dfrac{\partial L}{%
\partial y^{v}},
\end{eqnarray*}
thus

\begin{eqnarray*}
-\delta L_{c}&=&\left(\dfrac{d}{dt}\dfrac{\partial L_{c}}{\partial y^{\bar{u}}}-%
\dfrac{\partial L_{c}}{\partial x^{\bar{u}}}-C_{\bar{u}}^{u}\dfrac{\partial
L_{c}}{\partial x^{u}}\right)d x^{\bar{u}} \\
&=&\left(\dfrac{d}{dt}\dfrac{\partial L}{\partial y^{\bar{u}}}-\dfrac{\partial L}{%
\partial x^{\bar{u}}}+C_{\bar{u}}^{u}\left( \dfrac{d}{dt}\dfrac{\partial L}{%
\partial y^{u}}-\dfrac{\partial L}{\partial x^{u}}\right) +\dfrac{\partial L%
}{\partial y^{u}}\left( \dfrac{d}{dt}C_{\bar{u}}^{u}-\dfrac{\partial C^{u}}{%
\partial x^{\bar{u}}}-C_{\bar{u}}^{v}\dfrac{\partial C^{u}}{\partial x^{v}}%
\right)\right)dx^{\bar{u}} \\
&=&-\delta L+\left(\dfrac{\partial L}{\partial y^{u}}\left( \dfrac{\partial C_{%
\bar{u}}^{u}}{\partial y^{\bar{v}}}\dfrac{dy^{\bar{v}}}{dt}+\dfrac{\partial
C_{\bar{u}}^{u}}{\partial x^{\bar{v}}}y^{\bar{v}}+\dfrac{\partial C_{\bar{u}%
}^{u}}{\partial x^{v}}C^{v}-\dfrac{\partial C^{u}}{\partial x^{\bar{u}}}-C_{%
\bar{u}}^{v}\dfrac{\partial C^{u}}{\partial x^{v}}\right)\right)dx^{\bar{u}} \\
&=&\dfrac{\partial L}{\partial y^{u}}\left( \dfrac{\partial ^{2}C^{u}}{%
\partial y^{\bar{u}}\partial y^{\bar{v}}}\dfrac{dy^{\bar{v}}}{dt}+\dfrac{%
\partial C_{\bar{u}}^{u}}{\partial x^{\bar{v}}}y^{\bar{v}}+\dfrac{\partial
C_{\bar{u}}^{u}}{\partial x^{v}}C^{v}-\dfrac{\partial C^{u}}{\partial x^{%
\bar{u}}}-C_{\bar{u}}^{v}\dfrac{\partial C^{u}}{\partial x^{v}}\right)dx^{\bar{u}} .
\end{eqnarray*}
Thus, using (\ref{formpsc}), one  has 
\begin{equation}
-\delta L_{c}=\dfrac{\partial L}{\partial y^{u}}\left( \dfrac{\partial
^{2}C^{u}}{\partial y^{\bar{u}}\partial y^{\bar{v}}}\dfrac{dy^{\bar{v}}}{dt}-%
\left[ C_{V},\left[ C_{V},\dfrac{\partial }{\partial y^{\bar{u}}}\right] %
\right] ^{u}\right) dx^{\bar{u}}.  \label{eqLc01}
\end{equation}

On the other hand,

\begin{eqnarray*}
\dfrac{d}{dt}\dfrac{\partial L_{c}}{\partial y^{\bar{u}}}-\dfrac{\partial
L_{c}}{\partial x^{\bar{u}}}-C_{\bar{u}}^{u}\dfrac{\partial L_{c}}{\partial
x^{u}}&=&\dfrac{\partial ^{2}L_{c}}{\partial y^{\bar{u}}\partial y^{\bar{v}}}%
\dfrac{dy^{\bar{v}}}{dt}+\dfrac{\partial ^{2}L_{c}}{\partial y^{\bar{u}%
}\partial x^{\bar{v}}}y^{\bar{v}}-\dfrac{\partial L_{c}}{\partial x^{\bar{u}}%
}-C_{\bar{u}}^{u}\dfrac{\partial L_{c}}{\partial x^{u}} \\
&=&\dfrac{\partial ^{2}L_{c}}{\partial y^{\bar{u}}\partial y^{\bar{v}}}%
\dfrac{dy^{\bar{v}}}{dt}+F_{\bar{u}},
\end{eqnarray*}
thus 
\begin{equation}
-\delta L_{c}=\left(\dfrac{\partial ^{2}L_{c}}{\partial y^{\bar{u}}\partial y^{%
\bar{v}}}\dfrac{dy^{\bar{v}}}{dt}+F_{\bar{u}}\right)dx^{\bar{u}},  \label{eqLc02}
\end{equation}%
where 
\begin{equation}
F_{\bar{u}}=\dfrac{\partial ^{2}L_{c}}{\partial y^{\bar{u}}\partial x^{\bar{v%
}}}y^{\bar{v}}-\dfrac{\partial L_{c}}{\partial x^{\bar{u}}}-C_{\bar{u}}^{u}%
\dfrac{\partial L_{c}}{\partial x^{u}}.  \label{deffub}
\end{equation}

Comparing the relations (\ref{eqLc01}) and (\ref{eqLc02}), we obtain 
\begin{equation}
\left( \dfrac{\partial L}{\partial y^{u}}\dfrac{\partial ^{2}C^{u}}{\partial
y^{\bar{u}}\partial y^{\bar{v}}}-\dfrac{\partial ^{2}L_{c}}{\partial y^{\bar{%
u}}\partial y^{\bar{v}}}\right) \dfrac{dy^{\bar{v}}}{dt}-F_{\bar{u}}-\dfrac{%
\partial L}{\partial y^{u}}\left[ C_{V},\left[ C_{V},\dfrac{\partial }{%
\partial y^{\bar{u}}}\right] \right] ^{u}=0.  \label{eqsemisph}
\end{equation}%
Denote 
\begin{equation}
h_{\bar{u}\bar{v}}=\dfrac{\partial L}{\partial y^{u}}\dfrac{\partial
^{2}C^{u}}{\partial y^{\bar{u}}\partial y^{\bar{v}}}-\dfrac{\partial
^{2}L_{c}}{\partial y^{\bar{u}}\partial y^{\bar{v}}}.  \label{dehnonl}
\end{equation}%
It is easy to see that $h=\left( h_{\bar{u}\bar{v}}\right) $ gives a global
bilinear form in the fibers of $VN\mathcal{F=}\pi _{N\mathcal{F}}^{\ast }N%
\mathcal{F}$.

We have, by a straightforward computation,%
\begin{equation*}
h_{\bar{u}\bar{v}}=\dfrac{\partial ^{2}L}{\partial y^{\bar{u}}\partial y^{%
\bar{v}}}+\dfrac{\partial C^{u}}{\partial y^{\bar{u}}}\dfrac{\partial ^{2}L}{%
\partial y^{\bar{v}}\partial y^{u}}+\dfrac{\partial C^{u}}{\partial y^{\bar{v%
}}}\dfrac{\partial ^{2}L}{\partial y^{\bar{u}}\partial y^{u}}+\dfrac{%
\partial C^{u}}{\partial y^{\bar{u}}}\dfrac{\partial C^{v}}{\partial y^{\bar{%
v}}}\dfrac{\partial ^{2}L}{\partial y^{u}\partial y^{v}}.
\end{equation*}

Using the splitting ($C^{\prime \prime }$ at left or $D^{\prime \prime }$at
right) of  the exact sequence (\ref{exseq011}) given by Proposition \ref%
{prconsdoi}, one can  easily deduce an interpretation of $h$. Recall that the
hessian of $L$ is a bilinear form in the fibers of $VTM=\pi _{TM}^{\ast }TM$.

\begin{proposition}
\label{prdefh}The bilinear form $h$ has the form $h=\left( D^{\prime \prime
}\right) ^{\ast }H_{L}$, where $H_{L}$ is the vertical hessian of $L$,
restricted to $N\mathcal{F}$, as the image of the constraint map $C:N%
\mathcal{F}\rightarrow TM$.
\end{proposition}

\begin{proof} We use adapted coordinates. The conclusion follows using the
form (\ref{formDsec}) and the identity%
\begin{equation*}
\left( h_{\bar{u}\bar{v}}\right) =\left( 
\begin{array}{cc}
\delta _{\bar{u}}^{\bar{u}_{1}} & C_{\bar{u}}^{u_{1}}%
\end{array}%
\right) \left( 
\begin{array}{cc}
\dfrac{\partial ^{2}L}{\partial y^{\bar{u}_{1}}\partial y^{\bar{u}_{2}}} & 
\dfrac{\partial ^{2}L}{\partial y^{\bar{u}_{1}}\partial y^{u_{2}}} \\ 
\dfrac{\partial ^{2}L}{\partial y^{u_{1}}\partial y^{\bar{u}_{2}}} & \dfrac{%
\partial ^{2}L}{\partial y^{u_{1}}\partial y^{u_{2}}}%
\end{array}%
\right) \left( 
\begin{array}{c}
\delta _{\bar{v}}^{\bar{u}_{2}} \\ 
C_{\bar{v}}^{u_{2}}%
\end{array}%
\right) .
\end{equation*}%
\end{proof}

We say that the lagrangian $L$ is $C$\emph{--regular} if the bilinear form $h
$ is nondegenerated on the fibers of $VN\mathcal{F}$. If  that is  the case,
denoting 
\begin{equation*}
\left( h^{\bar{u}\bar{v}}\right) =\left( h_{\bar{u}\bar{v}}\right) ^{-1},
\end{equation*}%
the equation (\ref{eqsemisph}) gives 
\begin{eqnarray}
\dfrac{dy^{\bar{u}}}{dt} &=&S^{\bar{u}}\overset{def.}{=}{}h^{\bar{u}\bar{v}%
}\left( F_{\bar{v}}+\dfrac{\partial L}{\partial y^{u}}\left[ C_{V},\left[ C_{V},%
\dfrac{\partial }{\partial y^{\bar{v}}}\right] \right] ^{u}\right)   \notag
\\
&=&h^{\bar{u}\bar{v}}\left( \dfrac{\partial ^{2}L_{c}}{\partial y^{\bar{v}%
}\partial x^{\bar{w}}}y^{\bar{w}}-\dfrac{\partial L_{c}}{\partial x^{\bar{v}}%
}-C_{\bar{v}}^{u}\dfrac{\partial L_{c}}{\partial x^{u}}+\dfrac{\partial L}{%
\partial y^{u}}\left[ C_{V},\left[ C_{V},\dfrac{\partial }{\partial y^{\bar{v%
}}}\right] \right] ^{u}\right) .  \label{eqsemispS}
\end{eqnarray}

By a straightforward computation, based on the equality (\ref{eqsemisph}),
one can prove that the local functions $(S^{\bar{u}})$ verify the rule (\ref%
{formS}) on the intersection of compatible domains, giving by formula (\ref%
{locfS}) an almost transverse semi-spray $S$, called as \emph{canonically
associated} with $C$ and $L$. Using Proposition \ref{prextsemisptr} and the
above constructions, one have the following result.

\begin{proposition}
\label{prexs}If the lagrangian $L$ is $C$--regular, then the integral
curves, solutions of equations of motion of the extended nonholonomic
system, are the integral curves of a $C$--semispray $\mathcal{S}$.
\end{proposition}

Notice that in the particular case of linear and affine constraints, using
formulas (\ref{formpsc}), then (\ref{eqconst02}) can be deduced from (\ref%
{eqsemispS}).

The Legendre map of $L$, $\mathcal{L}:TM\rightarrow T^{\ast }M$, or $%
\mathcal{L}\in \Gamma (\pi _{TM}^{\ast }T^{\ast }M)$, is given in local
coordinates by%
\begin{equation*}
\mathcal{L}=\dfrac{\partial L}{\partial y^{u}}dy^{u}+\dfrac{\partial L}{%
\partial y^{\bar{u}}}dy^{\bar{u}}.
\end{equation*}

The statement below contains the main result of the paper. It gives all the
equations of motion in the same form as in \cite{BKMM, Bl}, all in presence
of constraints adapted to a regular foliation.

\begin{theorem}
\label{thmain}Let $L:TM\rightarrow I\!\!R$ be a lagrangian, $C:N\mathcal{F}%
\rightarrow TM$ be a nonlinear constraint and $L_{c}=L\circ C$ be the
constrained lagrangian. If one of the following conditions holds:

\begin{enumerate}
\item $L$ is $C$--regular, or

\item $C$ is a linear constraint, or

\item $C$ is an affine constraint
\end{enumerate}

then the constrained Lagrange equations (\ref{eqlagc}) have the form%
\begin{equation*}
\delta L_{c}=\left\langle R,\mathcal{L}\right\rangle ,
\end{equation*}%
or, using local coordinates,%
\begin{equation*}
\dfrac{d}{dt}\dfrac{\partial L_{c}}{\partial y^{\bar{u}}}-\dfrac{\partial
L_{c}}{\partial x^{\bar{u}}}+C_{\bar{u}}^{u}\dfrac{\partial L_{c}}{\partial
x^{u}}=\dfrac{\partial L}{\partial y^{u}}R_{\bar{u}}^{u},
\end{equation*}%
where, in the first case, $R$ is the $S$--curvature of $C$ and $S$ is the almost transverse
semi-spray canonically associated, and, in the last two cases, $R$ is the
curvature of $C$.
\end{theorem}

In the case of time dependent constraints (as in \cite{LiBe}, see Example \ref{e3}
in the next section), but a time independent lagrangian, the equations of
motion are obtained in the same way as equations (\ref{eqLc01}), taking into
account the fact that the constraints are time dependent, but the lagrangian
is not. One  obtains that the equations (\ref{eqLc01}) are replaced by 
\begin{equation}
-\delta L_{c}=\left( \dfrac{\partial ^{2}C^{u}}{\partial y^{\bar{u}}\partial
y^{\bar{v}}}\dfrac{dy^{\bar{v}}}{dt}\frac{\partial L}{\partial y^{u}}+\dfrac{%
\partial ^{2}C^{u}}{\partial t\partial y^{\bar{u}}}\dfrac{\partial L}{%
\partial y^{u}}-\dfrac{\partial L}{\partial y^{u}}\left[ C_{V},\left[ C_{V},%
\dfrac{\partial }{\partial y^{\bar{u}}}\right] \right] ^{u}\right) dx^{\bar{u%
}}.  \label{eqLc01a}
\end{equation}%
and the equations (\ref{eqsemisph}) are valid in the same form, but with%
\begin{equation}
F_{\bar{u}}=-\dfrac{\partial ^{2}C^{u}}{\partial t\partial y^{\bar{u}}}\dfrac{%
\partial L}{\partial y^{u}}-\dfrac{\partial ^{2}L_{c}}{\partial t\partial y^{%
\bar{u}}}+\dfrac{\partial ^{2}L_{c}}{\partial y^{\bar{u}}\partial x^{\bar{v}}%
}y^{\bar{v}}-\dfrac{\partial L_{c}}{\partial x^{\bar{u}}}-C_{\bar{u}}^{u}%
\dfrac{\partial L_{c}}{\partial x^{u}}.  \label{deffub-a}
\end{equation}

Let us consider the case when the foliation on $M$ is simple, given by a
fibered manifold $\pi :M\rightarrow M^{\prime }$. Then $h=\left( h_{\bar{u}%
\bar{v}}\right) $ gives a global bilinear form in the fibers of $VNM\mathcal{%
=}\pi _{NM}^{\ast }NM$; according to Proposition \ref{prdefh}, $h=\left(
D^{\prime \prime }\right) ^{\ast }H_{L}$, where $D^{\prime \prime }$ is the
right splitting map in the exact sequence (\ref{exseq011s}) and $H_{L}$ is
the vertical hessian of $L$, restricted to $NM$, as the image of the
constraint map $C:NM\rightarrow TM$. The condition that $L$ is $C$--regular
reads that $h$ is non-degenerate. In the simple foliation case, 2. and 3. of
Theorem \ref{thmain} are just the cases of linear and affine constraints
studied in \cite{BKMM, Bl}. But 1. of Theorem \ref{thmain} is new also in
the simple foliation case. It asserts in fact that if the lagrangian $L$ is $%
C$--regular, then the formal equations from \cite{BKMM, Bl} have the same
form in this case.

\section{Examples}

\begin{example}
\label{e1}
We consider the case of Appell's linear constraints.

The lagrangian is%
\begin{equation*}
L=\dfrac{1}{2}%
\alpha \left( \left( y^{1}\right) ^{2}+\left( y^{2}\right) ^{2}\right) +%
\dfrac{1}{2}\beta \left( y^{3}\right) ^{2}+\dfrac{1}{2}I_{1}\left( y^{\bar{1}%
}\right) ^{2}+\dfrac{1}{2}I_{2}\left( y^{\bar{2}}\right) ^{2}+\gamma x^{3}.
\end{equation*}%
and the constraints are given by (\ref{exlinAppC}). The induced lagrangian
has the form%
\begin{eqnarray*}
L_{c}(x^{1},x^{2},x^{3},x^{\bar{1}},x^{\bar{2}},y^{\bar{1}},y^{\bar{2}}) &=&%
\dfrac{1}{2}\left( I_{1}+\alpha R^{2}+\beta r^{2}\right) \left( y^{\bar{1}%
}\right) ^{2}+\dfrac{1}{2}I_{2}\left( y^{\bar{2}}\right) ^{2}+\gamma x^{3} \\
&=&\dfrac{1}{2}\alpha ^{\prime \prime }\left( y^{\bar{1}}\right) ^{2}+\dfrac{%
1}{2}I_{2}\left( y^{\bar{2}}\right) ^{2}+\gamma x^{3}.
\end{eqnarray*}

Using formulas (\ref{CurvApp}) and (\ref{eqcurvL}), we have%
\begin{eqnarray*}
R_{\bar{1}}^{1} &=&B_{\bar{1}\bar{2}}^{1}y^{\bar{2}}=-Ry^{\bar{2}}\sin x^{%
\bar{2}},\ R_{\bar{2}}^{1}=B_{\bar{2}\bar{1}}^{1}y^{\bar{1}}=Ry^{\bar{1}%
}\sin x^{\bar{2}}, \\
R_{\bar{1}}^{2} &=&B_{\bar{1}\bar{2}}^{2}y^{\bar{2}}=Ry^{\bar{2}}\cos x^{%
\bar{2}},\ R_{\bar{2}}^{2}=B_{\bar{2}\bar{1}}^{2}y^{\bar{1}}=-Ry^{\bar{1}%
}\cos x^{\bar{2}}, \\
R_{\bar{1}}^{3} &=&B_{\bar{1}\bar{2}}^{3}y^{\bar{2}}=0,R_{\bar{2}}^{3}=B_{\bar{2}\bar{1}}^{3}y^{\bar{1}}=0.
\end{eqnarray*}%
In this case $(h_{\bar{u}\bar{v}})$ given by formula (\ref{dehnonl}) %%
is  minus the  hessian of $L_{c}$. Formula (\ref{eqsemispS}) gives $S^{\bar{u}}=\dfrac{dy^{%
\bar{u}}}{dt}=-h^{\bar{u}\bar{v}}C_{\bar{v}}^{3}\dfrac{\partial L_{c}}{\partial x^{3}}%
$; by a straightforward computation, one  obtains 
\begin{equation*}
\dfrac{dy^{\bar{1}}}{dt}=-\alpha ^{\prime \prime }r,~\dfrac{dy^{\bar{2}}}{dt}%
=0.
\end{equation*}%
Since $y^{\bar{u}}=\dfrac{dy^{\bar{u}}}{dt}$, we obtain $x^{\bar{1}}=-\dfrac{%
r\alpha ^{\prime \prime }t^{2}}{2}+y_{0}^{\bar{1}}t+x_{0}^{\bar{1}}$, $x^{%
\bar{2}}=y_{0}^{\bar{2}}t+x_{0}^{\bar{2}}$; using constraint equations (\ref%
{exlinAppC}), one obtain $x^{1}=-R\left( r\alpha ^{\prime \prime }t+y_{0}^{%
\bar{1}}\right) \cos \left( y_{0}^{\bar{2}}t+x_{0}^{\bar{2}}\right)
,x^{2}=-R\left( r\alpha ^{\prime \prime }t+y_{0}^{\bar{1}}\right) \sin
\left( y_{0}^{\bar{2}}t+x_{0}^{\bar{2}}\right) ,C^{3}=-r\left( r\alpha
^{\prime \prime }t+y_{0}^{\bar{1}}\right) $.
\end{example}

\begin{example}
\label{e2}
We consider the case of Appell's nonlinear constraints,
with lagrangian as, for example, in \cite{LMD, ILMD}:%
\begin{equation}
L(x^{1},x^{\bar{1}},x^{\bar{2}},y^{1},y^{\bar{1}},y^{\bar{2}})=\dfrac{\beta 
}{2}\left( \left( y^{\bar{1}}\right) ^{2}+\left( y^{\bar{2}}\right)
^{2}\right) +\dfrac{\gamma }{2}\left( y^{1}\right) ^{2}+\delta x^{1}.
\label{Lag-AH}
\end{equation}%
and the constraints (\ref{con-nelApp}). The induced lagrangian has the form%
\begin{equation*}
L_{c}(x^{1},x^{\bar{1}},x^{\bar{2}},y^{\bar{1}},y^{\bar{2}})=\dfrac{\beta
+\alpha^2 \gamma }{2}\left( \left( y^{\bar{1}}\right) ^{2}+\left( y^{\bar{2}%
}\right) ^{2}\right) +\delta x^{1},
\end{equation*}%
thus formula (\ref{eqsemispS}) gives $S^{\bar{u}}=\dfrac{dy^{\bar{u}}}{dt}%
=-h^{\bar{u}\bar{v}}C_{\bar{v}}^{1}\dfrac{\partial L_{c}}{\partial x^{1}}$, where $%
(h^{\bar{u}\bar{v}})=(h_{\bar{u}\bar{v}})^{-1}$ and $(h_{\bar{u}\bar{v}})$
is given by formula (\ref{dehnonl}). By a straightforward computation, one obtains 
\begin{equation*}
\dfrac{dy^{\bar{u}}}{dt}=\dfrac{\alpha ^{\prime }y^{\bar{u}}}{\sqrt{\left(
y^{\bar{1}}\right) ^{2}+\left( y^{\bar{2}}\right) ^{2}}},
\end{equation*}%
where $\alpha ^{\prime }=-\dfrac{\alpha }{2\left( \gamma \alpha ^{2}+\beta
\right) }$. Using polar coordinates $y^{\bar{1}}=\rho \cos \varphi $, $y^{%
\bar{2}}=\rho \sin \varphi $, it follows that $\dfrac{d\rho }{dt}=\alpha $, $%
\dfrac{d\varphi }{dt}=0$, thus $\rho =\alpha ^{\prime }t+\rho _{0}$, $%
\varphi =\varphi _{0}$. Since $y^{\bar{u}}=\dfrac{dx^{\bar{u}}}{dt}$, one
have $x^{\bar{1}}=\left( \dfrac{\alpha ^{\prime }t^{2}}{2}+\rho _{0}t\right)
\cos \varphi _{0}+x_{0}^{\bar{1}}$, $x^{\bar{2}}=\left( \dfrac{\alpha
^{\prime }t^{2}}{2}+\rho _{0}t\right) \sin \varphi _{0}+x_{0}^{\bar{2}}$, $%
x^{1}=\pm \alpha \left( \dfrac{\alpha ^{\prime }t^{2}}{2}+\rho _{0}t\right)
+x_{0}^{1}$. The solutions of the constrained Lagrange equations are
straight lines; but this is not physically correct in the case of Appell
machine (see \cite{Ma}).
\end{example}

\begin{example}
\label{e3}
In the Appell-Hammel dynamic system in an elevator,
 the constraints are (\ref{ex-AH}) and one takes the lagrangian (\ref{Lag-AH}), as in Example \ref{e2}; 
thus Example \ref{e2} is a particular case of this example, when $%
v^{0}(t)=0$. Using Proposition \ref{primpl}, we can infer at this stage
that, concerning the solution of the equation of motion, one obtains the same
result as in \cite[Section 3.2]{LiBe}. Indeed, we have 
\begin{equation*}
L_{c}=\dfrac{\beta +\alpha ^{2}\gamma }{2}\left( \left( y^{\bar{1}}\right)
^{2}+\left( y^{\bar{2}}\right) ^{2}\right) +\gamma v^{\left( 0\right) }\sqrt{%
\left( y^{\bar{1}}\right) ^{2}+\left( y^{\bar{2}}\right) ^{2}}+\delta x^{1}+%
\dfrac{1}{2}\left( v^{\left( 0\right) }\right) ^{2},
\end{equation*}%
then we have that the pseudo-curvature $R_{V}=\dfrac{\partial L}{\partial
y^{u}}\left[ C_{V},\left[ C_{V},\dfrac{\partial }{\partial y^{\bar{u}}}%
\right] \right] ^{u}=0$ and 
\begin{eqnarray*}
\left( F_{\bar{u}}\right)  &=&\left( -\dfrac{\left( \alpha \delta +\gamma 
\dot{v}^{(1)}\right) y^{\bar{u}}}{\sqrt{\left( y^{\bar{1}}\right)
^{2}+\left( y^{\bar{2}}\right) ^{2}}}\right) , \\
\left( \dfrac{\partial L}{\partial y^{1}}\dfrac{\partial ^{2}C^{1}}{\partial
y^{\bar{u}}\partial y^{\bar{v}}}-\dfrac{\partial ^{2}L_{c}}{\partial y^{\bar{%
u}}\partial y^{\bar{v}}}\right) ^{-1}\left( F_{\bar{u}}\right) ^{t}
&=&\left( \dfrac{\alpha \left( \alpha \delta +\gamma \dot{v}^{(1)}\right) y^{%
\bar{v}}}{\alpha ^{2}\gamma +\beta }\right) ,
\end{eqnarray*}%
thus formula (\ref{eqsemispS}) gives%
\begin{equation*}
\ddot{x}^{\bar{u}}=\dot{y}^{\bar{u}}=\dfrac{\left( \alpha \delta +\gamma 
\dot{v}^{(1)}\right) y^{\bar{u}}}{\sqrt{\left( y^{\bar{1}}\right)
^{2}+\left( y^{\bar{2}}\right) ^{2}}},
\end{equation*}%
then%
\begin{equation*}
\dfrac{d}{dt}\sqrt{\left( \dot{x}^{\bar{1}}\right) ^{2}+\left( \dot{x}^{\bar{%
2}}\right) ^{2}}=\dfrac{\ddot{x}^{\bar{1}}\dot{x}^{\bar{1}}+\ddot{x}^{\bar{2}%
}\dot{x}^{\bar{2}}}{\sqrt{\left( \dot{x}^{\bar{1}}\right) ^{2}+\left( \dot{x}%
^{\bar{2}}\right) ^{2}}}=\left( \alpha \delta +\gamma \dot{v}^{(1)}\right) =%
\dfrac{d}{dt}\left( \alpha \delta t +\gamma v^{(1)}\right) 
\end{equation*}%
and one obtains the solution as in \cite[Section 3.2]{LiBe}.
\end{example}

\begin{example}
\label{e4}
The following example fits in the case of equations of
Benenti mechanism \cite{Ben} (see also \cite{KO}). Consider the foliation of 
$I\!\!R^{4}$ with coordinates $(x_{1},x_{2},y_{1},y_{2})$ generated by $%
\dfrac{\partial }{\partial x_{1}}$. Denote $x_{2}=x^{\bar{1}}$, $y_{1}=x^{%
\bar{2}}$, $y_{2}=x^{\bar{3}}$ and $x_{1}=x^{1}$ and consider the nonlinear
constraint given by the implicit equation $y^{1}y^{\bar{3}}-y^{\bar{1}}y^{%
\bar{2}}=0$, $\left( y^{1}\right) ^{2}+\left( y^{\bar{1}}\right) ^{2}+\left(
y^{\bar{2}}\right) ^{2}+\left( y^{\bar{3}}\right) ^{2}\neq 0$, where $%
(x^{1},x^{\bar{1}},x^{\bar{2}},x^{\bar{3}},y^{1},y^{\bar{1}},y^{\bar{2}},y^{%
\bar{3}})$ are coordinates on $TI\!\!R^{4}$. We have $C^{1}\left( x^{1},x^{%
\bar{1}},x^{\bar{2}},x^{\bar{3}},y^{\bar{1}},y^{\bar{2}},y^{\bar{3}}\right) =%
\dfrac{y^{\bar{1}}y^{\bar{2}}}{y^{\bar{3}}}$. Formula (\ref{formDsec}) gives 
\begin{equation*}
X^{1}\dfrac{\partial }{\partial x^{1}}+X^{\bar{1}}\dfrac{\partial }{\partial
x^{\bar{1}}}+X^{\bar{2}}\dfrac{\partial }{\partial x^{\bar{2}}}+X^{\bar{3}}%
\dfrac{\partial }{\partial x^{\bar{3}}}\overset{C^{\prime }}{\rightarrow }%
{}\left( X^{1}+\dfrac{X^{\bar{1}}y^{\bar{2}}y^{\bar{3}}+X^{\bar{2}}y^{\bar{1}%
}y^{\bar{3}}-X^{\bar{3}}y^{\bar{1}}y^{\bar{2}}}{\left( y^{\bar{3}}\right)
^{2}}\right) \dfrac{\partial }{\partial x^{1}}.
\end{equation*}%
 It can be seen that $R_{V}=0$, only using euclidean coordinates.

 Let us consider the lagrangian%
\begin{equation*}
L(x^{1},x^{\bar{1}},x^{\bar{2}},y^{1},y^{\bar{1}},y^{\bar{2}})=\dfrac{\alpha 
}{2}\left( \left( y^{1}\right) ^{2}+\left( y^{\bar{1}}\right) ^{2}\right) +%
\dfrac{\beta }{2}\left( \left( y^{\bar{2}}\right) ^{2}+\left( y^{\bar{3}%
}\right) ^{2}\right) +f(x^{1},x^{\bar{1}},x^{\bar{2}},x^{\bar{3}}),
\end{equation*}%
that has the kinetic energy as the original \cite{Ben}, or as in \cite{ILMD,
KO}. The induced lagrangian has the form%
\begin{equation*}
L_{c}(x^{1},x^{\bar{1}},x^{\bar{2}},y^{\bar{1}},y^{\bar{2}})=\dfrac{\alpha
\left( y^{\bar{1}}\right) ^{2}}{2\left( y^{\bar{3}}\right) ^{2}}\left(
\left( y^{\bar{2}}\right) ^{2}+\left( y^{\bar{3}}\right) ^{2}\right) +\dfrac{%
\beta }{2}\left( \left( y^{\bar{2}}\right) ^{2}+\left( y^{\bar{3}}\right)
^{2}\right) +f(x^{1},x^{\bar{1}},x^{\bar{2}},x^{\bar{3}}).
\end{equation*}%
Formula (\ref{deffub}) gives $F_{\bar{1}}=-f,_{\bar{1}}+f,_{1}\dfrac{y^{\bar{%
2}}y^{\bar{3}}}{\left( y^{\bar{1}}\right) ^{2}}$, $F_{\bar{2}}=-f,_{\bar{2}%
}-f,_{1}\dfrac{y^{\bar{3}}}{y^{\bar{1}}}$, $F_{\bar{3}}=-f,_{\bar{3}}-f,_{1}%
\dfrac{y^{\bar{2}}}{y^{\bar{1}}}$, where $f,_{1}=\dfrac{\partial f}{\partial
x^{1}}$ and $f,_{\bar{u}}=\dfrac{\partial f}{\partial x^{\bar{u}}}$. Then 
\begin{equation*}
\left( h^{\bar{u}\bar{v}}\right) =\left( 
\begin{array}{ccc}
{\scriptsize -}\tfrac{\left( y^{\bar{1}}\right) ^{2}\left( \left( y^{\bar{1}%
}\right) ^{2}\alpha \beta +\left( y^{\bar{1}}\right) ^{2}\beta ^{2}+\left(
y^{\bar{2}}\right) ^{2}\alpha ^{2}+\left( y^{\bar{2}}\right) ^{2}\alpha
\beta +\left( y^{\bar{3}}\right) ^{2}\alpha \beta \right) }{\left( y^{\bar{2}%
}y^{\bar{3}}\right) ^{2}\alpha \beta \left( \alpha +\beta \right) } & 
{\scriptsize -}\tfrac{y^{\bar{1}}}{y^{\bar{2}}\left( \alpha +\beta \right) }
& {\scriptsize -}\tfrac{y^{\bar{1}}}{y^{\bar{3}}\beta } \\ 
{\scriptsize -}\tfrac{y^{\bar{1}}}{y^{\bar{2}}\left( \alpha +\beta \right) }
& {\scriptsize -}\tfrac{1}{\alpha +\beta } & {\scriptsize 0} \\ 
{\scriptsize -}\tfrac{y^{\bar{1}}}{y^{\bar{3}}\beta } & {\scriptsize 0} & 
{\scriptsize -}\tfrac{1}{\beta }%
\end{array}%
\right) 
\end{equation*}%
and formula (\ref{eqsemispS}) gives 
\begin{equation*}
S^{\bar{u}}=\dfrac{dy^{\bar{u}}}{dt}=h^{\bar{u}\bar{v}}F_{\bar{v}}.
\end{equation*}%
In the Benenti original example \cite{Ben}, the potential $f$ vanishes, then 
$S^{\bar{u}}=0$, thus the integral curves are straight lines.
\end{example}

\begin{example}
\label{e5}
The following example is the Marle servomechanism \cite%
{Ma0} (see also \cite{KO}), where the Chetaev principle is claimed to fail
in the real world. Consider the foliation of $I\!\!R^{2}$ with coordinates $%
(x^{1},x^{\bar{1}})$ generated by $\dfrac{\partial }{\partial x^{1}}$ and
consider the nonlinear constraint given by 
\begin{equation*}
y^{1}=f(x^{1},x^{\bar{1}},y^{\bar{1}}).
\end{equation*}%
We have $C^{3}\left( x^{1},x^{\bar{1}},y^{\bar{1}}\right) =f(x^{1},x^{\bar{1}%
},y^{\bar{1}})$. Formula (\ref{formDsec}) gives 
\begin{equation*}
X^{1}\dfrac{\partial }{\partial x^{1}}+X^{\bar{1}}\dfrac{\partial }{\partial
x^{\bar{1}}}\overset{C^{\prime }}{\rightarrow }{}\left( X^{1}+X^{\bar{1}}%
\dfrac{\partial f}{\partial y^{\bar{1}}}\right) \dfrac{\partial }{\partial
x^{1}}.
\end{equation*}%
 Then $R_{V}=$ $\left[ \dfrac{\partial }{\partial y^{\bar{1}}},\left[ 
\dfrac{\partial }{\partial y^{\bar{1}}},f\dfrac{\partial }{\partial x^{1}}%
+y^{\bar{1}}\dfrac{\partial }{\partial x^{\bar{1}}}\right] \right] =\dfrac{%
\partial ^{2}f}{\partial \left( y^{\bar{1}}\right) ^{2}}\dfrac{\partial }{%
\partial x^{1}}$.

 One can consider the lagrangian%
\begin{equation*}
L(x^{1},x^{\bar{1}},y^{1},y^{\bar{1}})=\dfrac{m}{2}\left( \left(
y^{1}\right) ^{2}-2ly^{1}y^{\bar{1}}\sin x^{\bar{1}}\right) +\dfrac{J}{2}%
\left( y^{\bar{1}}\right) ^{2}-mgl\sin x^{\bar{1}},
\end{equation*}%
thus%
\begin{equation*}
L_{c}(x^{1},x^{\bar{1}},y^{\bar{1}})=\dfrac{m}{2}\left( f^{2}-2lfy^{\bar{1}%
}\sin x^{\bar{1}}\right) +\dfrac{J}{2}\left( y^{\bar{1}}\right) ^{2}-mgl\sin
x^{\bar{1}},
\end{equation*}

If one  considers the
induced foliation (generated by $\dfrac{\partial }{\partial x^{1}}$) 
on $V=I\!\!R^{2}\backslash \{(x^{1},0)|x^{1}\geq 0\}$ , it can be seen that
this is not a locally trivial one and the space of leaves is not Hausdorff
separated, thus the use of a foliation is justified in this case from a
global viewpoint.
\end{example}

\begin{example}
\label{e6}
Consider a riemannian flow, i.e. a one-dimensional
riemannian foliation (see \cite{Mo}), on a manifold $M$. It means that the following
are given on $M$:

\begin{description}
\item- a one-dimensional distribution (given, for example, by a non-vanishing
vector field $X_{0}$) and

\item- a bundle-like metric $g$.
\end{description}

Using local coordinates $(x^{1},x^{\bar{u}})$ on $M$, adapted to the
foliation, the vector field $X_{0}$ and the quadratic lagrangian of the
metric $g$ have the local forms 
\begin{eqnarray*}
X_{0} &=&X^{0}\left( x^{1},x^{\bar{u}}\right) \dfrac{\partial }{\partial
x^{1}}, \\
L\left( x^{1},x^{\bar{u}},y^{1},y^{\bar{u}}\right) &=&\dfrac{1}{2}%
g_{0}\left( x^{1},x^{\bar{u}}\right) ^{2}\left( y^{1}\right) ^{2}+\dfrac{1}{2%
}g_{\bar{u}\bar{v}}\left( x^{\bar{u}}\right) y^{\bar{u}}y^{\bar{v}}.
\end{eqnarray*}%
We notice that $g_{0}$ can be taken $1$ only in the case when the orthogonal
distribution to the riemannian flow is integrable.

Let us consider a quadratic time-dependent constraint of the form $%
L_{0}\left( x^{1},x^{\bar{u}},y^{1},y^{\bar{u}}\right) =\dfrac{1}{2}\varphi
\left( t\right) $.

It is easy to see that the lagrangian induced on the constraint manifold is $%
L_{c}=\dfrac{1}{2}\varphi \left( t\right) $. It is completely degenerated,
and we show how the constraint machinery works in this case. We have%
\begin{equation*}
C\left( x^{1},x^{\bar{u}},y^{\bar{u}}\right) =\dfrac{1}{g_{0}}\sqrt{\varphi
\left( t\right) -g_{\bar{u}\bar{v}}\left( x^{\bar{u}}\right) y^{\bar{u}}y^{%
\bar{v}}}=\dfrac{g}{g_{0}}.
\end{equation*}%
As usually in the presence of a metric, denote $y_{\bar{u}}=g_{\bar{u}\bar{v}%
}y^{\bar{v}}$, thus $g_{\bar{u}\bar{v}}\left( x^{\bar{u}}\right) y^{\bar{u}%
}y^{\bar{v}}=y_{\bar{v}}y^{\bar{v}}$. By a straightforward computation, we
obtain%
\begin{equation*}
\dfrac{\partial C}{\partial y^{\bar{u}}}=-\dfrac{y_{\bar{u}}}{gg_{0}},\dfrac{%
\partial ^{2}C}{\partial t\partial y^{\bar{u}}}=\dfrac{\varphi ^{\prime
}\left( t\right) y_{\bar{u}}}{2g^{3}g_{0}},\;\dfrac{\partial ^{2}C}{\partial
y^{\bar{u}}\partial y^{\bar{v}}}=-\dfrac{y_{\bar{u}}y_{\bar{v}}+g^{2}g_{\bar{u%
}\bar{v}}}{g_{0}g^{3}}.
\end{equation*}%
The matrix $\left( \dfrac{\partial ^{2}C}{\partial y^{\bar{u}}\partial y^{%
\bar{v}}}\right) $ has as inverse 
\begin{equation*}
\left( \dfrac{\partial ^{2}C}{\partial y^{\bar{u}}\partial y^{\bar{v}}}%
\right) ^{-1}=-\left( g_{0}gg^{\bar{u}\bar{v}}-\dfrac{g_{0}gy^{\bar{u}}y^{%
\bar{v}}}{\varphi }\right) .
\end{equation*}

The formulas (\ref{eqcurvat}), (\ref{deffub-a}) and (\ref{eqsemisph}) give
in this case:
\begin{eqnarray*}
&&R_{\bar{u}}^{1}=C\dfrac{\partial ^{2}C}{\partial x^{1}\partial y^{\bar{u}}}%
+y^{\bar{v}}\dfrac{\partial ^{2}C}{\partial x^{\bar{v}}\partial y^{\bar{u}}}-%
\dfrac{\partial C}{\partial x^{\bar{u}}}-\dfrac{\partial C}{\partial y^{\bar{%
u}}}\dfrac{\partial C}{\partial x^{1}},\\
&&F_{\bar{u}}=-\dfrac{\partial ^{2}C}{\partial t\partial y^{\bar{u}}}\dfrac{%
\partial L}{\partial y^{1}}-\dfrac{\partial ^{2}L_{c}}{\partial t\partial y^{%
\bar{u}}}+\dfrac{\partial ^{2}L_{c}}{\partial y^{\bar{u}}\partial x^{\bar{v}}%
}y^{\bar{v}}-\dfrac{\partial L_{c}}{\partial x^{\bar{u}}}-C_{\bar{u}}^{1}%
\dfrac{\partial L_{c}}{\partial x^{1}}=-\dfrac{\partial ^{2}C}{\partial
t\partial y^{\bar{u}}}\dfrac{\partial L}{\partial y^{1}}=-\dfrac{%
Cg_{0}\varphi ^{\prime }\left( t\right) y_{\bar{u}}}{2g^{3}},\\
&&\left( \dfrac{\partial L}{\partial y^{1}}\dfrac{\partial ^{2}C}{\partial y^{%
\bar{u}}\partial y^{\bar{v}}}-\dfrac{\partial ^{2}L_{c}}{\partial y^{\bar{u}%
}\partial y^{\bar{v}}}\right) \dfrac{dy^{\bar{v}}}{dt}-F_{\bar{u}}+\dfrac{%
\partial L}{\partial y^{1}}R_{\bar{u}}^{1}=Cg_0^2\dfrac{\partial ^{2}C}{\partial
y^{\bar{u}}\partial y^{\bar{v}}}\dfrac{dy^{\bar{v}}}{dt}+\dfrac{%
Cg_{0}\varphi ^{\prime }\left( t\right) y_{\bar{u}}}{2g^{3}}+Cg_0^2 R_{\bar{u}%
}^{1}=0.
\end{eqnarray*}

Thus 
\begin{equation*}
\dfrac{dy^{\bar{v}}}{dt}=-\left( g_0gg^{\bar{u}\bar{v}}-\dfrac{g_0gy^{\bar{u}}y^{%
\bar{v}}}{\varphi }\right) \left(- \dfrac{\varphi ^{\prime }\left(
t\right) y_{\bar{u}}}{2g^{3}g_0}-R_{\bar{u}}^{1}\right) ,
\end{equation*}%
or 
\begin{equation*}
\dfrac{dy^{\bar{v}}}{dt}=\dfrac{y^{\bar{v}}\varphi ^{\prime }}{2\varphi 
}+g_0g\left( g^{\bar{u}\bar{v}}-\dfrac{y^{\bar{u}}y^{\bar{v}}}{%
\varphi }\right) R_{\bar{u}}^{1}.
\end{equation*}%
In the case when the curvature coefficients vanish, we have%
\begin{equation*}
\dfrac{dy^{\bar{v}}}{dt}=\dfrac{y^{\bar{v}}\varphi ^{\prime }}{2\varphi 
}.
\end{equation*}%
It follows that $y^{\bar{v}}=c_{\bar{v}}\sqrt{\varphi }$. Since $y^{\bar{v}}=\dfrac{dx^{\bar{v}}}{dt}$, it follows
that the trajectories have the local form%
\begin{equation*}
x^{\bar{v}}=c_{\bar{v}}\Phi \left( t\right) +d_{\bar{v}},
\end{equation*}%
where $c_{\bar{v}}$ and $d_{\bar{v}}$ are constants and $\Phi \left(
t\right) \in \int \sqrt{\varphi \left( t\right) }dt$, i.e. a primitive of $%
\sqrt{\varphi \left( t\right)}$.

This applies in the particular case of a decelerated motion of a free
particle, studied initially in \cite{Kr} and then in \cite{Sw}, where the
explicit form of the trajectories are obtained. In this case $\varphi \left(
t\right) =\frac{1}{t}$, when $\Phi =2\sqrt{t}$.

Let us notice that almost the same computations hold in the case of a
pseudo-riemannian metric. In particular, when the signature of $g$ is $%
(1,m-1)$ and the pseudo-riemannian flow is in the positive direction, then the
computations are similar, the initial data having some changes of signs.
\end{example}

\section*{Acknowledgement} The authors cordially thank to anonymous referees for several useful remarks and suggestions about the initial submission which improve  this paper.

\noindent Paul Popescu\newline
Department of Applied Mathematics, University of Craiova\newline
Address: Craiova, 200585, Str. Al. Cuza, No. 13, Rom\^{a}nia\newline
email:\textit{paul$_{-}$p$_{-}$popescu@yahoo.com} \medskip

\noindent Cristian Ida\newline
Department of Mathematics and Informatics, University Transilvania of Bra%
\c{s}ov\newline
Address: Bra\c{s}ov 500091, Str. Iuliu Maniu 50, Rom\^{a}nia\newline
email:\textit{cristian.ida@unitbv.ro}

\end{document}